\documentclass[letterpaper,twocolumn,10pt]{article}
\usepackage{usenix-2020-09}

\usepackage{titling}

\date{}

\usepackage{nopageno}
\usepackage[english]{babel}
\usepackage{blindtext}

\usepackage{tikz}
\usepackage{amsmath}

\usepackage{algorithm}
\usepackage{algpseudocode}

\usepackage{amsthm}
\usepackage{multirow}
\usepackage{color}
\usepackage{times}
\usepackage{tabularx}
\usepackage{tabulary}
\usepackage{subcaption}

\usepackage{microtype} % To save space
\usepackage{enumitem}
\usepackage{graphbox}
\usepackage{listings}
\definecolor{commentgreen}{rgb}{0.18, 0.55, 0.34}
\usepackage{tikz}

\def\Snospace~{\S{}}

\newtheorem{theorem}{Theorem}

\usepackage{comment}
\usepackage{xspace}
\begin{document}

\lstdefinelanguage{json}{
    % numbers=left,
    % numberstyle=\small,
    % frame=single,
    rulecolor=\color{black},
    % showspaces=false,
    % showtabs=false,
    % breaklines=true,
    postbreak=\raisebox{0ex}[0ex][0ex]{\ensuremath{\color{green}\hookrightarrow\space}},
    % breakatwhitespace=true,
    % basicstyle=\ttfamily\small,
    upquote=true,
    morestring=[b]",
    stringstyle=\color{commentgreen},
    literate=
     *{0}{{{\color{red}0}}}{1}
      {1}{{{\color{red}1}}}{1}
      {2}{{{\color{red}2}}}{1}
      {3}{{{\color{red}3}}}{1}
      {4}{{{\color{red}4}}}{1}
      {5}{{{\color{red}5}}}{1}
      {6}{{{\color{red}6}}}{1}
      {7}{{{\color{red}7}}}{1}
      {8}{{{\color{red}8}}}{1}
      {9}{{{\color{red}9}}}{1}
      {\{}{{{\color{blue}{\{}}}}{1}
      {\}}{{{\color{blue}{\}}}}}{1}
      {[}{{{\color{blue}{[}}}}{1}
      {]}{{{\color{blue}{]}}}}{1},
}

\newenvironment{denseitemize}{
\begin{itemize}[topsep=2.5pt, partopsep=0pt, leftmargin=1.5em]
  \setlength{\itemsep}{2.5pt}
  \setlength{\parskip}{0pt}
  \setlength{\parsep}{0pt}
}{\end{itemize}}

\newcounter{packednmbr}
\newenvironment{denseenum}{
\begin{list}{\thepackednmbr.}{\usecounter{packednmbr}
\setlength{\itemsep}{0pt}
\addtolength{\labelwidth}{4pt}
\setlength{\leftmargin}{12pt}
\setlength{\listparindent}{\parindent}
\setlength{\parsep}{1pt}
\setlength{\topsep}{0pt}}}
{\end{list}
}

\newcommand{\redx}{\color{BrickRed} \ding{55}}
\newcommand{\greencheck}{\color{OliveGreen} $\checkmark$}

\newcommand{\todo}[1]{\textbf{\color{red}TODO: #1}}
\newcommand{\cut}[1]{}

\newcommand{\name}[0]{DAK}
\newcommand{\taskname}[0]{service}

% To make the personalized notes go away, comment out these lines...

% \newcommand{\shouxu}[1]{\textbf{\color{orange}SX: #1}}
% \newcommand{\jxlin}[1]{\textbf{\color{blue}JX: #1}}
% \newcommand{\zy}[1]{{\textbf{\color{green}ZY: #1}}}

% ...and uncomment these lines.
% \newcommand{\fixme}[1]{}
\newcommand{\shouxu}[1]{}
\newcommand{\jxlin}[1]{}
\newcommand{\zy}[1]{}

\newcommand{\ratio}[2]{#1\,:\,#2}
\newcommand{\scaption}[1]{\caption{\footnotesize \textbf{#1}}}
\newcommand{\scaptionl}[2]{\caption{\footnotesize \textbf{#1}#2}}

\newcommand{\mypara}[1]{\smallskip\noindent{\bf {#1}}~}
\newcommand{\ie}[0]{\textit{i.e.,}\xspace}
\newcommand{\eg}[0]{\textit{e.g.,}\xspace}
\newcommand{\etc}[0]{\textit{etc.,}\xspace}
\newcommand{\vs}[0]{\textit{vs.}\xspace}
\newcommand{\etal}[0]{\textit{et al.}\xspace}
% \conferenceinfo{HotNets 2021} {}
% \CopyrightYear{2021}
% \crdata{X}
% \date{}

%%%%%%%%%%%% THIS IS WHERE WE PUT IN THE TITLE AND AUTHORS %%%%%%%%%%%%
% \titlebanner{DRAFT---Do not distribute}

% \title{\vspace{-0.7in}Enabling Portable and High-Performance SmartNIC Programs with Alkali}
\title{\name{}: Direct-Access-Enabled GPU Memory Offloading with Optimal Efficiency for LLM Inference }
\date{\vspace{-1.8em}}

% Usenix Authors

% \author{Under submission}

% \preauthor{\vspace{-5em}}

% Usenix Authors
\author{
{\rm Shouxu Lin}\\
Cornell University
\and
{\rm Zhiyuan Guo}\\
Cornell University
\and
{\rm Jiaxin Lin}\\
Cornell University
}

% \end author

\maketitle

\newcommand\blfootnote[1]{%
  \begingroup
  \renewcommand\thefootnote{}\footnote{#1}%
  \addtocounter{footnote}{-1}%
  \endgroup
}

% \vspace{-}
%-------------------------------------------------------------------------------
\begin{abstract}
LLM inference is constrained by GPU memory capacity and bandwidth. Tiered memory architectures mitigate this by allowing the GPU to offload memory to the remote tier. However, existing memory offloading frameworks rely on prefetching data into local GPU HBM. This approach underutilizes system resources by introducing HBM contention, squandering memory capacity, and creating pipeline bubbles. We show that enabling direct GPU access to remote memory significantly outperforms prefetching, achieving optimal aggregate system bandwidth. We propose \name{}, an end-to-end direct-access memory offloading framework that repurposes the Tensor Memory Accelerator (TMA) to asynchronously fetch offloaded weights and KV caches directly from remote memory into GPU shared memory (SMEM). To maximize remote access performance, \name{} introduces a greedy algorithm to determine optimal per-operation offloading ratios, alongside active congestion control and TMA multicast to eliminate interconnect bottlenecks and read amplification. Evaluations across diverse architectures show that \name{} achieves near-optimal bandwidth aggregation, with up to 3$\times$ performance gains on NVLink-C2C and 1.8$\times$ on PCIe systems compared to state-of-the-art memory offloading baselines.
\end{abstract}

% \vspace{.3cm}
% \begingroup\small\noindent\raggedright\textbf{PVLDB Artifact Availability:}\\
% The source code, data, and/or other artifacts h.
% \endgroup

% \begingroup
% \renewcommand\thefootnote{}\footnote{\noindent
% \textbf{\color{red}{This work is under submission, please do not distribute.}}
% }\endgroup

% For example, in the Nvidia Grace Hopper (GH200) architecture, allowing the GPU to access tightly integrated CPU memory via NVLink-C2C increases the GPU's effective memory capacity by 500\% and expands its total available memory bandwidth by 12.5\%. \jxlin{I think GPU+ CXL can also increase aggregate bandwidth, need to check}. We are seeing an increasing number of industry pushes in next-generation GPU in adopting tired memory architecture to scale memory capacity and bandwidth~\cite{}.

\section{Introduction}
\label{sec:intro}

Large Language Models (LLMs) are increasingly starved for memory capacity and bandwidth due to massive parameter scaling, ultra-long contexts, and large inference batch sizes. Scaling the GPU's local High Bandwidth Memory (HBM) to meet this demand is prohibitively expensive and fundamentally constrained by die area and thermal limits~\cite{gholami2024aimemorywall}. To breach this "memory wall," recent architectures adopt \textit{tiered memory designs}. By connecting the GPU to remote memory (e.g., host CPU memory via NVLink-C2C or CXL)~\cite{li2025fenghuang,fusco2024understanding, cxlconsortium2023spec,nvidia2023gh200arch}, these systems decouple compute and memory scaling. For instance, the NVIDIA Grace Hopper (GH200)~\cite{nvidia2023gh200arch} utilizes NVLink-C2C to connect the GPU with the CPU memory, which increases the effective memory capacity by 500\% and bandwidth by 12.5\%.

\begin{figure}[t!]
    \centering
    \includegraphics[width=0.65\linewidth]{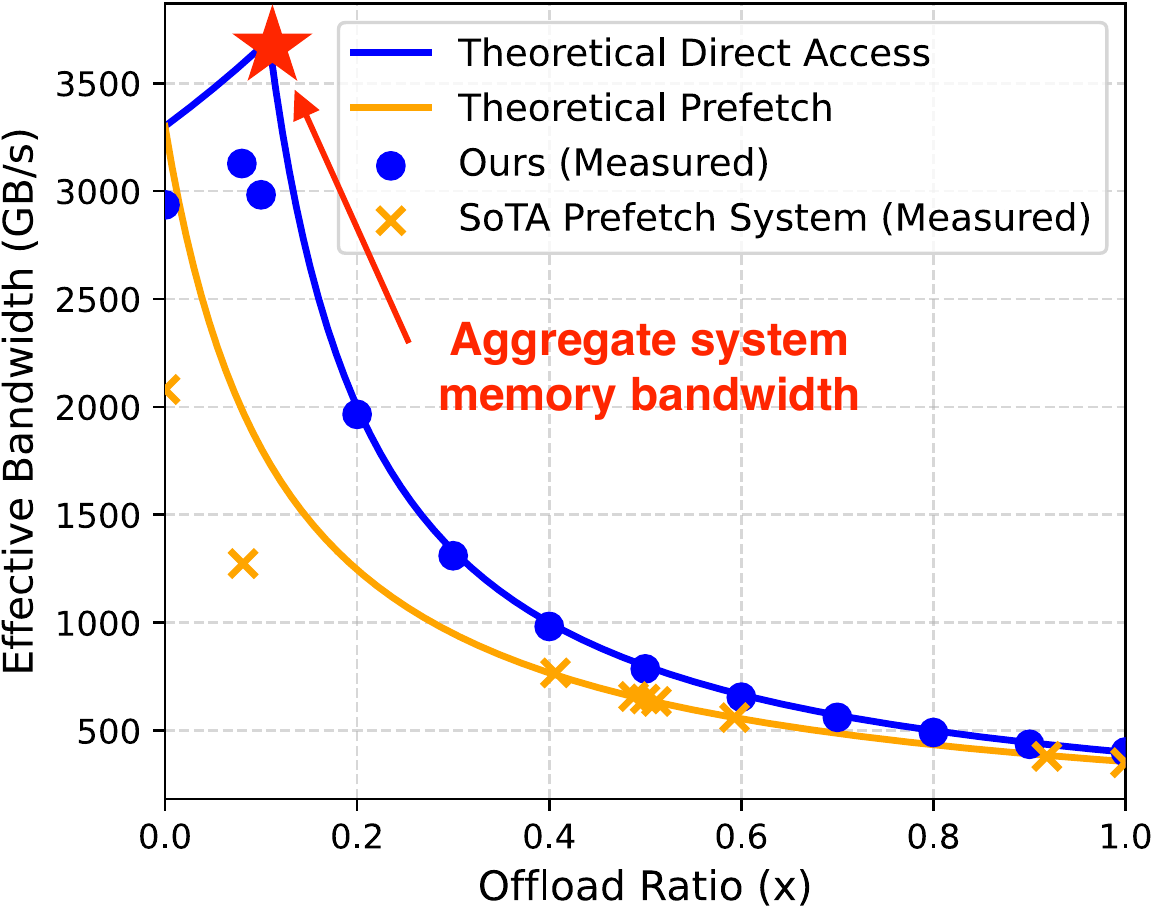}
    \vspace{-0.2cm}
    \scaption{Comparison of a prefetch-based system and our direct-access-based system, measured on GH200 with OPT-30b.}
     \vspace{-0.4cm}
    \label{fig:intro}
\end{figure}

The central challenge in tiered memory is exploiting remote capacity and bandwidth without suffering interconnect latency penalties. To overcome the challenge of utilizing remote capacity without suffering interconnect latency penalties, prefetching-based frameworks have been proposed. In the context of LLM inference, these frameworks offload a subset of weights and KV caches in remote memory and hide the remote data fetching latencies, with aid of computation-aware prefetching~\cite{sheng2023flexgen, kamahori2024fiddler, xu2024pie, vijaya2025aqua, jiang2025neo, kim2025lia, cao2025moe},  by analyzing the execution graph and preloading required inputs (weights, and KV caches) from remote memory into GPU HBM before a specific computation block begins. 

However, we observe that prefetching fundamentally underutilizes both the bandwidth and capacity of tiered memory systems. This bandwidth and capacity underutilization occurs for two reasons. First, prefetching fails to exploit remote bandwidth and leads to degraded local GPU memory bandwidth, as the incoming prefetch traffic (moving data from remote memory to HBM) competes directly with active compute kernels for HBM bandwidth. Second, prefetching limits maximum system memory capacity by forcing data to be staged in local GPU HBM before computation, consuming extra bounce buffers that reduce GPU HBM capacity. Third, prefetching-based approaches are usually coarse-grained (e.g., layer-based~\cite{sheng2023flexgen,xu2024pie,kim2025lia}), which seldom achieve perfect overlap between data transfer and computation, leading to pipeline bubbles and low efficiency. Consequently, as shown in \autoref{fig:intro}, the theoretical performance bound of prefetching falls significantly short of the ideal aggregate system bandwidth (i.e., local GPU HBM plus the host-GPU interconnect). Furthermore, real-world prefetching systems suffer an additional 20\% performance degradation due to pipeline bubbles caused by imperfect compute-communication overlap.

 We argue that enabling the GPU to directly access remote memory can yield significantly better performance than prefetching data through HBM,  by allowing GPU Streaming Multi-processors (SMs) to fetch offloaded data directly into shared memory (SMEM). Directly allowing GPU SM’s to access remote memory, we can bypass HBM staging entirely and fully utilize aggregate system bandwidth and capacity, thus reducing both memory and memory bandwidth wastage. In \autoref{fig:intro}, we show that this direct access mechanism can provide a strictly superior theoretical and empirical performance compared to prefetching.

To enable direct access to remote memory, we propose \name{} (Direct Access Kernel). Traditionally, a direct data path required GPU threads to explicitly issue load/store instructions over the interconnect, consuming valuable SM register space and instruction cycles. \name{} overcomes this by utilizing the Tensor Memory Accelerator (TMA)~\cite{nvidia2022hopper}, a specialized hardware engine introduced in the Hopper architecture. We observe that, TMA, which was originally designed to manage asynchronous transfers between local HBM and SMEM, can be repurposed to execute asynchronous, direct accesses from remote host memory straight into SMEM. \textit{To the best of our knowledge, we are the first to utilize the TMA for remote access and memory offload.}

However, identifying this hardware capability is only the first step, and we found that a naive implementation of TMA-based remote access is insufficient and can yield non-optimal performance. To make direct remote access achieve near-theoretical maximum performance, we must solve three critical system-level challenges:
1) \textbf{Determine Offload Ratio:} Ignoring the characteristics of memory- and compute-bound kernels in the inference pipeline and blindly offloading memory can underutilize the system. An analytical algorithm is required to determine the optimal offloading ratio for each specific kernel. 
2) \textbf{Interconnect Congestion:} We observe that unconstrained in-flight TMA accesses to host memory can trigger severe intra-GPU interconnect congestion and stall the GPU's local HBM accesses.
3) \textbf{Uncacheable Host Memory:} CPU memory always bypasses the GPU L2 cache (even on coherent architectures like GH200), and uncacheable host memory can cause huge read amplification over interconnects.

To solve these challenges, we propose \name{}, an end-to-end direct-access memory offloading framework. \name{} optimally partitions weights and KV caches across host/GPU memory tiers and employs TMA-enabled warp-specialized kernels to decouple memory fetching from matrix math and efficiently hides heterogeneous memory latencies. \name{} overcomes challenges and achieves performance near the theoretical maximum (\autoref{fig:intro}) with contributions:

\begin{denseenum}
    \item \textbf{Direct Memory Offload Architecture:} We develop warp-specialized \textit{SplitK\_GEMM} and \textit{SplitK\_FlashAttn} kernels that compute directly over tiered memory and leverage TMA to fetch offloaded data from remote memory to SMEM. This achieves true bandwidth aggregation, expands memory capacity, and hides remote memory latencies through in-kernel compute-communication overlap.
    \item \textbf{Optimal Per-operation Offload Ratio:} We formulate the memory offloading ratio as a performance optimization problem. We design an efficient greedy algorithm that analytically determines the optimal partitioning ratios for individual compute-bound and memory-bound operations, with provable optimality.
    \item \textbf{Efficient TMA Access to Host Memory:} We identify and mitigate the congestion and read-amplification bottlenecks of naive TMA remote access. We implement active congestion control to prevent host memory requests from stalling local HBM, and we utilize a TMA multicast and host-locality-first tile scheduling to eliminate read amplification for uncacheable host data.
    \item \textbf{Cross-Architecture Generality:} Our direct-access kernels can be adapted across varying interconnect bandwidths (e.g., NVLink-C2C vs. PCIe), distinct GPU architectures, and diverse model configurations, while also ensuring robust performance.
\end{denseenum} 

We provide a comprehensive evaluation of \name{}'s end-to-end effectiveness. \name{} is implemented with an accessible PyTorch interface and evaluated across two distinct hardware environments: an NVLink-C2C-based GH200 and a PCIe-based RTX 6000 Blackwell system. We test across multiple models (OPT and Llama) under various batch sizes and sequence length configurations. Compared to multiple state-of-the-art, prefetching-based memory offloading systems, our results show that \name{} consistently improves performance across all tested configurations and architectures, delivering up to a 3$\times$ higher throughput on the GH200 and a 1.8$\times$ gain on the RTX 6000 Pro Blackwell. \name{} is fully open-sourced at \url{https://github.com/shouxulin/DirectAccessKernel.git}.

\section{Background and Motivation}
\label{sec:motiv}

\subsection{Scaling Memory with Tiered Architecture}
The rapid evolution of LLM places unprecedented pressure on \textit{memory capacity} and \textit{bandwidth} in inference~\cite{kim2024breakthrough,wolters2024memory,alizadeh2024llm}. 

\noindent\textbf{Memory Capacity bottleneck:} 
Modern LLMs require massive memory capacity for both model weights and KV-Cache. For example, a 671B-parameter model in FP16 requires 1.34 TB of memory for weights alone. At the same time, the KV cache scales linearly with both batch size and sequence length, and can easily exceed 100 GB for a 70B model serving a 100K context at a batch size of 16, rapidly exhausting the GPU’s local HBM capacity.

% LLM inference workloads rely heavily on the GPU's memory capacity to accommodate both exploding model weights. For example, a 671B-parameter model in FP16 requires over 1.34 TB of memory just for the weights. Beyond weights, the KV-Cache size also posts large requirements on capacity. KV-Cache scales linearly with batch size and context length. With ultra-long context inference (e.g., 100K+ tokens) and high-throughput, large-batch serving, it quickly drain up GPU memory, for example, serving a 70B-parameter model with a 100K context window at a batch size of just 16 requires over 100 GB of memory exclusively to store the KV-Cache.

\noindent\textbf{Memory Bandwidth bottleneck:} 
LLM inference requires ultra-high memory bandwidth. For example, during decoding, the entire model weights and the accumulated KV-Cache must be loaded into the computational units for every generated token. This results in low arithmetic intensity, where inference performance is bounded by memory bandwidth\cite{kim2025lia,agrawal2024taming}.

% LLM inference has distinct characteristics of the two inference phases: prefill and decode. While the initial prefill phase (processing the input prompt) is generally compute-intensive, the subsequent decode phase (generating tokens one at a time) is strictly memory-bound. During decoding, the entire model footprint and the accumulated KV-Cache must be loaded from memory into the computational units for every single generated token. This results in extremely low arithmetic intensity, leaving the GPU's compute units severely starved for data if the underlying memory bandwidth is insufficient.
\begin{figure}[t!]
    \centering
    \includegraphics[width=0.99\linewidth]{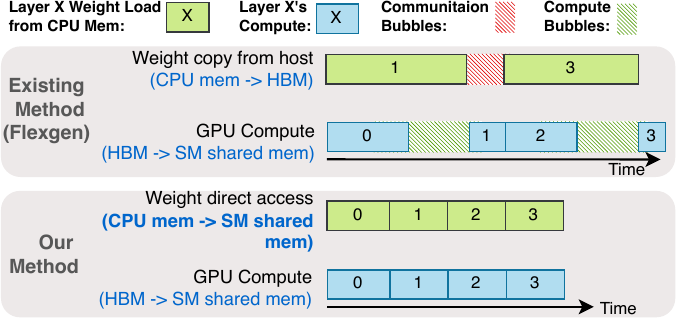}
    \vspace{-0.2cm}
    \scaption{Comparing the copy-based method with our method. In FlexGen, layer 1 and 3 weights are offloaded to the CPU.}
     \vspace{-0.2cm}
    \label{fig:flexgen}
\end{figure}

\noindent\textbf{Scaling Memory with Tiered Architecture:}
% Despite the increasing demand for memory capacity and bandwidth, scaling local High Bandwidth Memory (HBM) directly on the GPU die is fundamentally constrained by strict thermal envelopes, power limitations, and packaging reticle limits~\cite{}. 
An effective strategy to breach this "memory wall" is to allow accelerators access to external memory. This is enabled by the tiered memory, where the GPU is connected to a remote memory tier such as host CPU memory and CXL memory pool~\cite{li2023pond,gouk2023memory}. This enables memory resources to scale independently of GPU compute, providing large capacity expansions.
% with great cost-efficiency.

This tiered paradigm is also driven by rapid advancements in high-speed compute-memory interconnects like NVLink-C2C~\cite{fusco2024understanding,schieffer2024harnessing}, die-to-die interconnects \cite{feng2023heterogeneous}, and AMD Infinity Fabric~\cite{pearson2023interconnect}. For example, in the Grace Hopper (GH200) architecture, the NVLink-C2C connects CPU and GPU with a bandwidth of 450 GB/s per direction, which is close to the host CPU memory bandwidth (500 GB/s). NVLink-C2C increases the GPU's effective capacity by 500\% and expands the total available system memory bandwidth by 12.5\% \footnote{Total available system memory bandwidth is defined by $GPU\_HBM\_BW + MIN(NVLink\_C2C\_BW, HOST\_DRAM\_BW)$}.

\subsection{Memory Offloading}
The central research question in tiered memory architectures is how to use remote memory capacity and bandwidth without falling victim to interconnect latency and throughput constraints. Prior work focuses exclusively on using remote memory merely as a capacity pool, managing the data by explicitly copying it between remote memory and GPU local HBM.

\noindent\textbf{Paging with Prefetching:} Existing works propose general-purpose GPU memory management systems that access remote memory using paging abstractions \cite{sheng2023flexgen,kamahori2024fiddler,kim2025lia,xu2024pie}. These systems rely on prefetching to bring remote pages into the GPU's local memory before they are needed. Frameworks such as CUDA Managed Memory~\cite{li2015evaluation,schieffer2024harnessing} implement prefetching algorithms that copy memory pages from the remote tier into local GPU HBM prior to kernel execution. These prefetching decisions can be triggered either explicitly through programmer-controlled APIs (e.g., cudaMemPrefetchAsync) or implicitly via software-~\cite{allen2021depth} or hardware-managed prefetchers~\cite{schieffer2024harnessing}.

\noindent\textbf{LLM Pipeline Orchestration:}
Another line of work builds customized GPU kernels or inference pipelines that explicitly control when data is moved from remote memory into local GPU HBM. These systems frame memory offloading as a scheduling optimization problem to maximize the overlap between remote memory communication and active GPU computation. For example, FlexGen~\cite{sheng2023flexgen} formulates a linear programming solver to orchestrate model weights and KV cache placement across the GPU, CPU, and disk. To overlap communication with computation, it uses layer-by-layer prefetching. As shown in \autoref{fig:flexgen} (1), while the GPU executes computation for layer $X$, it concurrently fetches the weights for layer $X+1$ from CPU memory to GPU HBM.

Systems like Neo~\cite{jiang2025neo} and Lia~\cite{kim2025lia} jointly perform memory and computation offloading. Computation offloading is typically beneficial when: (1) operations are strictly memory-bound, making the slower CPU a reasonable alternative, and (2) interconnect bandwidth (\eg PCIe) is significantly lower than host DRAM bandwidth, and by offloading computation to host can reduce costly data movement across interconnect. However, as next-generation interconnects (\eg NVLink-C2C in Grace Hopper and Blackwell) achieve bandwidths matching CPU DRAM, the benefits of host-side computation offloading diminish~\cite{kamahori2024fiddler,kim2025lia}.

% Aqua \cite{} leverages NVLink to expand the memory across multi-GPU domains via memory copies and similarly schedules the timing for memory copies to hide interconnect latencies.

\begin{figure}[t!]
    \centering
    \includegraphics[width=0.99\linewidth]{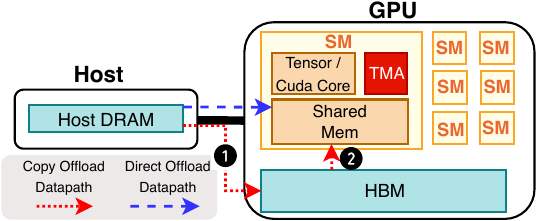}
    % \vspace{-0.2cm}
    \scaption{Direct memory offload vs. copy data paths.}
     \vspace{-0.2cm}
    \label{fig:directoffload}
\end{figure}

\subsection{Limitation of Copy-Based Offloading}
\label{subsec:limitation}
We observe that all existing copy-based memory offloading solutions are fundamentally suboptimal in utilizing both the bandwidth and capacity of remote memory.

\noindent \textbf{Underutilized System Memory Bandwidth.}
Existing copy-based solutions degrade local GPU memory bandwidth rather than aggregating total system bandwidth. As shown in \autoref{fig:directoffload}, the standard datapath requires copying data from the host to HBM (\textcircled{1}), and then reading it from HBM to SM shared memory (\textcircled{2}). Because HBM is not fully bidirectional, the incoming host writes (\textcircled{1}) contend directly with active SM reads (\textcircled{2}). In memory-bound kernels, this contention degrades throughput; for example, prefetching layer 1 weights directly interferes with layer 0 computation (\autoref{fig:flexgen}). This degradation scales with the interconnect-to-local bandwidth ratio. 
Background copies can degrade local memory performance by up to 10\% on a GH200, and by roughly 4-7\% on PCIe-based consumer-grade GPU (64 GB/s PCIe Gen4 vs. 960 GB/s GDDR7)~\cite{nvidia2025rtx50series, nvidia2025rtxpro_blackwell_whitepaper}.

\noindent\textbf{Communication Bubbles and Wasted HBM Capacity.}
Copy-based prefetching trade-off between capacity and interconnect utilization. To overlap communication with computation, systems must allocate static HBM staging buffers. As shown in \autoref{fig:flexgen} (1), a naive double-buffered prefetcher will encounters communication bubbles: the interconnect idles while waiting for layer 1's computation to finish so its buffer can be reused for layer 3. While systems can mask these bubbles by increasing the prefetch depth (e.g., fetching $N$ layers ahead), doing so requires proportionally larger staging buffers, reducing HBM capacity for saving KV cache and weights.

\noindent\textbf{Computation Bubbles and Latency Trade-offs.} 
Layer-based coarse-grained prefetching often causes computation bubbles because executing a layer is much faster than fetching its weights (\autoref{fig:flexgen} (1)). To mask these stalls, systems like FlexGen process multiple micro-batches sequentially on the same layer to maximize weight reuse. However, this throughput optimization severely penalizes latency; forcing micro-batches to synchronize at each layer dramatically increases both time-to-first-token and per-token latencies.

% (paging, swapping, prefetching (ETC (ASPLOS 2019), FLEXGEN, TSPLIT (ICDE 2022))) 

\subsection{Memory Offload with Direct Access}
To achieve optimal performance, we need for a \textit{direct remote-memory-to-SM} data path (\autoref{fig:directoffload}). Direct access streams data directly from remote memory into the Streaming Multiprocessor's shared memory (SMEM), bypassing GPU HBM entirely. Traditionally, direct access required GPU threads to explicitly issue load/store instructions over the interconnect, consuming valuable register space and instruction cycles. We overcome this challenge by using the Tensor Memory Accelerator (TMA), a hardware engine introduced after the NVIDIA Hopper architecture. While originally designed for asynchronous transfers between local HBM and SMEM, we novelly prove (\autoref{sec:eval}) that TMA can execute asynchronous, direct accesses to remote memory with high-performance.

As shown in \autoref{fig:flexgen}, this direct-access scheme overcomes the limitations of copy-based methods through three critical benefits: (1) By streaming remote data directly into SMEM, SMs concurrently access host memory and local HBM to achieve the theoretical peak aggregated bandwidth.
(2) Eliminating static HBM staging buffers preserves GPU memory for larger models or KV caches.
(3) Instead of coarse-grained prefetching, direct access could use asynchronous TMA copies to seamlessly overlap fine-grained kernel computation with remote memory fetches, removing bubbles.

% As shown in \autoref{fig:flexgen}(2), direct access provides three advantages over copy-based methods:
% (1) Bandwidth Aggregation: By streaming remote data directly into SMEM, SMs concurrently access host memory and local HBM to achieve the theoretical peak of \textit{HBM + Interconnect Bandwidth}.
% (2) Capacity Savings: Eliminating static HBM staging buffers preserves GPU memory for larger models or KV caches.
% (3) Latency Hiding: Fine-grained, asynchronous TMA copies seamlessly overlap remote fetches with computation, eliminating the pipeline bubbles inherent in coarse-grained prefetching.
% \end{denseitemize}

\begin{figure}[t!]
    \centering
    \includegraphics[width=0.99\linewidth]{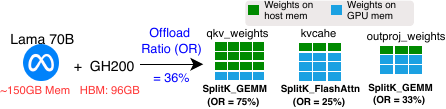}
    % \vspace{-0.2cm}
    \scaption{\name{} overview.}
     \vspace{-0.2cm}
    \label{fig:overview}
\end{figure}

\section{Overview}
\label{sec:overview}
We build \name{} to achieve highly efficient, direct-access memory offloading. The high-level idea of \name{} is simple: \textit{to fully utilize both local and remote memory, the GPU's SM should concurrently direct access HBM and host DRAM. This concurrent access is achieved through fine-grained weight partitioning and optimized GPU kernel running on partitioned memory.}

\autoref{fig:overview} shows a high-level overview of \name{}. When a user attempts to run a model on a given architecture, memory offloading is required if the available GPU HBM is smaller than the total memory footprint of the model's weights and KV cache. For example, running a 140 GB model like LLaMA-70B with batch size 32 and sequence length 1024 on a Grace Hopper (GH200) node with 96 GB of HBM requires a global memory offload ratio of 40\%. Given this global offload ratio, \name{} will calculate the optimal offloading ratio for each individual operation inside the inference pipeline\footnote{In this paper, we refer to operations involving KV cache matrices as \emph{attention} (\eg attention scoring and value aggregation in the attention layer), and operations involving model weights as \emph{linear} (\eg the q/k/v/o projections and the gate/up/down projection in MLPs). When we refer to partitioning an operation, we mean partitioning the input matrices of that operation.}. Based on these calculated ratios, each operation is partitioned and stored across both the GPU and host memory tiers.
After partitioning the operations, \name{} invokes optimized, direct-access kernels (\eg SplitK\_GEMM, SplitK\_FlashAttn) to execute over the split matrices. These kernels concurrently access local and remote memory. \name{} implements a warp-specialized producer-consumer architecture utilizing TMA accelerator for each kernel. This design decouples local/remote memory fetching from matrix math, allowing the GPU SMs to overlap computation with memory latency.

% \name{} overcomes the limitations discussed in \autoref{subsec:limitation}: 1.  GPU kernels achieve higher memory bandwidth for memory-bound operations (\eg decoding) by concurrently utilize both remote interconnect and local HBM bandwidth,. 2. Offloaded data is consumed directly via the SM, and no data needs to be staged in precious HBM. 3. Offloading and communication/computation overlapping are now performed at the sub-kernel (tile) level rather than the layer level.

% \subsection{Design Challenges and Ideas}
Realizing this direct-access architecture introduces many system-level challenges:
% Below, we describe these challenges and outline how we addresses them:

% \jxlin{COnsider C0, how to enable a kernel run with direct access on hetergenous memory}
\noindent\textbf{C1: Optimal Per-Operation Offloading Ratio.} How should the system determine the size of data offloaded to host memory for each individual operation, given a global offloading constraint? Different operations exhibit distinct bottlenecks (compute-bound vs. memory-bound), meaning each operation's performance has a different sensitivity to memory offloading. In \ref{subsec:offloading_algo}, we formulate this as an optimization problem and design an efficient greedy algorithm that determines per-operation offloading ratios with provable optimality.

\noindent\textbf{C2: Managing TMA-based Remote Memory Access.} In \name{}, GPU SMs concurrently access local and remote memory using the TMA. However, we observe that uncoordinated concurrent accesses can interfere with one another, causing severe intra-GPU interconnect congestion and limiting aggregate bandwidth. Furthermore, CPU memory is non-cacheable by the GPU's L2 cache (even under the recent cache-coherent Grace Hopper Architecture), requiring the kernel to handle disparate data paths and memory characteristics. In \ref{subsec:congestion_control}, we present a congestion control and TMA multicast mechanism that mitigates intra-GPU congestion and handles non-cacheable host memory access.

% \noindent\textbf{C3: Cross-Architecture Portability.} \name{} must seamlessly adapt to different GPU architectures, varying interconnect bandwidths (\textit{e.g.}, NVLink-C2C vs. PCIe), and diverse model sizes. In \autoref{sec:auto_tunner}, we introduce a cross-architecture auto-tuner that performs the aforementioned optimizations by profiling and adapting to the specific hardware characteristics of the deployment environment.

\section{Design}
\label{sec:design}

% \autoref{subsec:splitkernel} introduces how SplitKernel executes matrix computations across partitioned weights. \autoref{subsec:offloading_algo} explains how \name{} calculates optimal per-kernel offloading ratios. \autoref{subsec:congestion_control} introduces how to optimize SM-to-host memory access with congestion control and TMA-multicast. \autoref{subsec:kernel_align} discusses additional kernel optimizations and cross-architecture portability.

% \begin{itemize}
%     \item Offloading algorithm 
%     \begin{itemize}
%         \item Operators are heterogeneous in terms of performance bottleneck.
        
%         Some operators could be memory-bound and some could be compute-bound. In addition, even for computation-bound kernels, some could be more computation-bound than others. 

%         Arithmetic intensity is the factor that determines the performance bottleneck, which we use to characterize different kernels. 

%         \item Since operators are heterogeneous, applying a uniform offloading is sub-optimal. 

%     \end{itemize}
% \end{itemize}

\subsection{\name{} Design}
\label{subsec:splitkernel}

LLM inference is dominated by matrix multiplication operations, primarily involving the multiplication of hidden states by either model weights or KV cache matrices, in terms of both latency and memory footprint. \name{} partitions these operations across both GPU and host memory. Next, we introduce this partitioning scheme and the specialized kernels that compute directly over these partitioned matrices.

% LLM inference is dominated by matrix multiplication operations between hidden state tensors and model weights or KV cache tensors\footnote{In this paper, we refer to operations involving model weights as \emph{linear operations} (\eg q/k/v/o projections in the attention layer and gate/up/down projections in MLPs), and operations involving KV cache tensors as \emph{attention operations} (\eg attention score and value aggregation in the attention layer).}. Because these matrix multiplications dictate execution time, and the combination of model weights and KV cache largely determines the memory footprint, \name{} partitions these matrices across GPU and host memory. We develop specialized kernels that compute directly over these partitioned matrices, allowing the system to leverage both higher aggregate bandwidth and extended memory capacity.

% \begin{figure}[t!]
%     \centering
%     \includegraphics[width=0.99\linewidth]{figures/gemv-example.png}
%     \vspace{-0.2cm}
%     \scaption{An example of how \name{} partitions and executes a GEMV operation.}
%      \vspace{-0.2cm}
%     \label{fig:gemv_example}
% \end{figure}

\begin{figure}[t!]
    \centering
    \begin{subfigure}[t]{0.48\textwidth}
        \centering
        \includegraphics[clip,width=\textwidth]{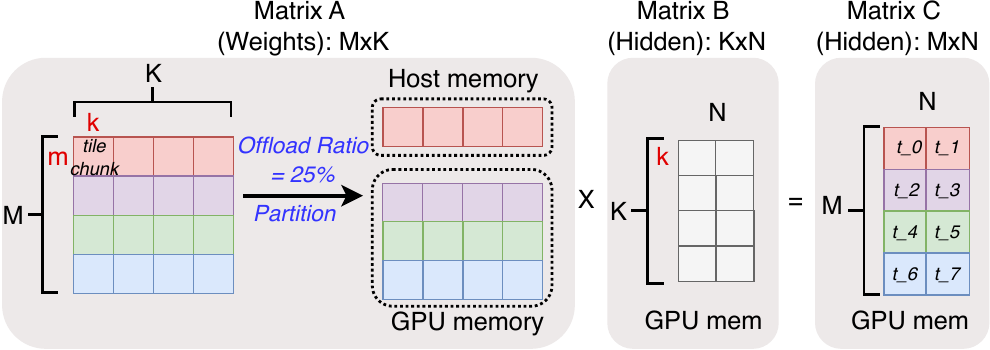}
        \scaption{Matrices partitioned between host and GPU memory.}
        \label{subfig:weights}
    \end{subfigure}%

    \begin{subfigure}[t]{0.48\textwidth}
        \centering
        \includegraphics[clip,width=\textwidth]{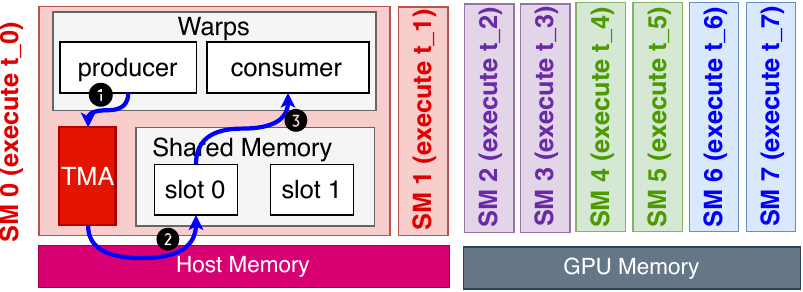}
        \scaption{Kernel on partitioned memories; tile rows executed by different SMs.}
        \label{subfig:GEMV}
    \end{subfigure}%
~
    \vspace{-0.2cm}
    \scaption{\name{} partitions and executes a matrix multiplication.}
    \vspace{-0.2cm}
    \label{fig:gemv_example} 
\end{figure}

\noindent\textbf{Data Partition:}
Figure~\ref{fig:gemv_example} illustrates how \name{} partitions a matrix multiplication operation, defined as $C = A \times B$, where $C$ and $B$ represent hidden states, and $A$ represents the model weights or KV cache. As shown in \autoref{subfig:weights}, $A$ is partitioned along the $M$ dimension into multiple tile rows, each of shape $m \times K$. 
These tile rows are then distributed across host memory and GPU memory. In the example shown, tile row $0$ resides in host memory, while tile rows $1$ through $3$ reside in GPU memory. The optimal number of tiles to be placed in host memory is determined by the offloading algorithm (detailed in \autoref{subsec:offloading_algo}).

\noindent\textbf{Kernel over Partitioned Memory:}
\name{} performs matrix computation directly over this partitioned memory. \autoref{subfig:GEMV} shows this kernel design. Each tile chunk in the output matrix is calculated by a different GPU SM. If the total number of tiles exceeds the available SM count, each SM processes multiple tile chunks. In this example, SMs 0 and 1 handle the computation for output tiles $t_0$ and $t_1$ respectively by reading weights from CPU host memory, while SMs 2 through 7 compute output tiles $t_2$ through $t_7$ by reading weights from local GPU memory. This design strictly isolates the data path, guaranteeing that each individual SM reads exclusively from either host or GPU memory.

The number SMs assigned to handle host versus GPU tiles is determined by three factors: (1) the \textbf{target offloading ratio}, which dictates the total size of data that must be processed from each memory tier; (2) \textbf{execution wave alignment}, which ensures the number of assigned tiles evenly divides across the allocated SMs to prevent tail latencies caused by partial, unbalanced execution waves; and (3) \textbf{interconnect congestion control}, which caps the number of host-assigned SMs when necessary to avoid memory contention (detailed in \autoref{subsec:congestion_control}).
% \jxlin{Let's make the kernel alignment brief, as we don't have space. Move it here.PTAL}

\noindent\textbf{Asynchronous TMA Access:}
The primary challenge of this approach is efficiently hiding the data load latencies that arise when SMs fetch matrices from HBM or remote host memory. To overlap data movement with computation, \name{} adopts a warp-specialized producer-consumer design within each SM, leveraging the Tensor Memory Accelerator (TMA). Since TMA load and store execute asynchronously, data movement happens in the background while Tensor Cores and CUDA cores perform mathematical operations.

As shown in \autoref{subfig:GEMV}, thread blocks in each SM utilize this warp specialization to mask memory latencies. The execution pipeline proceeds as follows: First, the producer warp leverages the TMA to asynchronously fetch a tile chunk of $A$ from host memory and its corresponding chunk of $B$ from GPU memory into a shared memory (SMEM) slot (Steps \textcircled{1} and \textcircled{2}). Once the slot is filled, the consumer warp performs the matrix multiplication (Step \textcircled{3}). Concurrently, the producer issues the next TMA requests to fetch the subsequent chunks into an alternate SMEM slot. The maximum number of in-flight TMA requests is determined by the available slot number and the congestion control algorithm introduced in \autoref{subsec:congestion_control}. By maintaining multiple in-flight TMA requests, the kernel effectively overlaps memory loads with computation.

With this design, \name{} achieves two critical benefits. First, it enables true parallel execution across memory tiers; by having distinct SMs read from host and GPU memory in parallel, the system fully achieve the aggregate bandwidth of both the local HBM and the host-GPU interconnect. Second, it efficiently absorbs the severe latency penalties of remote data loads through its asynchronous, warp-specialized nature.

\subsection{ Offloading Algorithm}
\label{subsec:offloading_algo}

This section introduces our offloading algorithm, which optimally distributes a global offloading budget across operations in LLM inference pipeline. We first demonstrate that applying a uniform offloading ratio across operations is suboptimal. To address this, we design a greedy offloading algorithm with \textit{provable optimality}.

\subsubsection{The Suboptimality of Uniform Offloading}
Given a global offloading ratio $OR$ (determined by total memory footprint and available GPU memory), the system must distribute this budget across individual operations. A naive approach is to apply a \textit{uniform offloading ratio}, where every operation offloads exactly $OR$ fraction of its involving matrices to the host. But this strategy is suboptimal because inference consists of a mixture of memory-bound and compute-bound operations, with different sensitivities to offloading.

\noindent\textbf{Diverse Performance Bottlenecks:}
An operation is memory-bound if its arithmetic intensity (AI) falls below the hardware's peak machine balance, and compute-bound otherwise. 
In LLM inference, the performance bottleneck—whether memory- or compute-bound—varies across operations and depends on workload characteristics such as batch size ($B$) and sequence length ($L$).  Both prefilling and decoding consist of a mixture of memory-bound and compute-bound operations:

\begin{denseitemize} 
\item \textbf{Decoding:} The attention operation is strictly memory-bound, as both its computation and memory traffic scale as $O(B \cdot L \cdot D_h)$ ( D\_h is the head dimension), yielding a constant Arithmetic Intensity (AI) of $O(1)$. In contrast, linear operations can be either memory- or compute-bound; it transitions from memory to compute-bound as $B$ grows.
    
\item \textbf{Prefilling:} The attention operation can be either memory- or compute-bound: its computation scales quadratically with sequence length ($O(B \cdot L^2 \cdot D_h)$) while memory scales linearly ($O(B \cdot L \cdot D_h)$), yielding an AI of $O(L)$ that shifts to compute-bound for long sequences. Linear operations achieve an AI of $O(B \cdot L)$ and remain inherently compute-bound since $B \cdot L$ is large in typical prefilling workloads.

\end{denseitemize}

\begin{figure}[t!]
    \centering
    \includegraphics[width=0.8\linewidth]{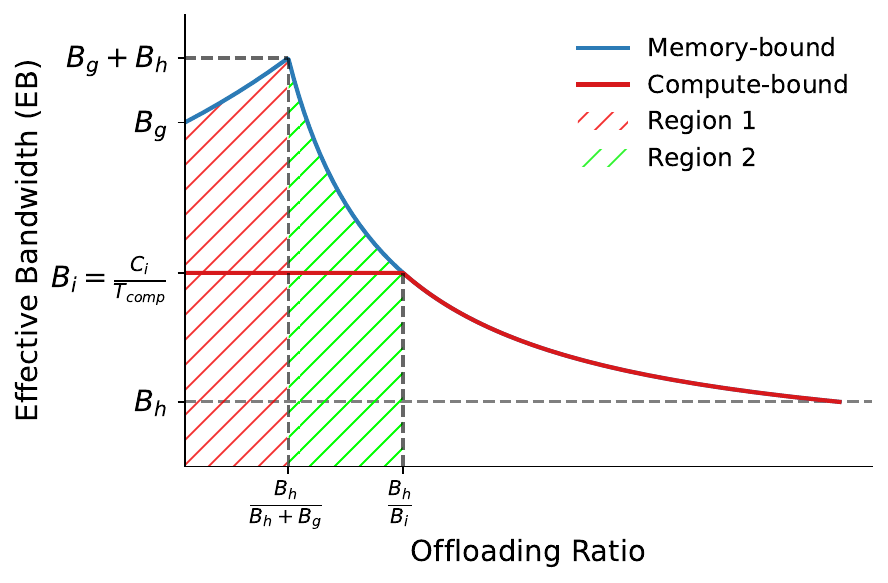}
    \vspace{-0.2cm}
    \scaption{Effective bandwidth of memory/compute-bound operations. }
     \vspace{-0.4cm}
    \label{fig:effective_bandwidth}
\end{figure}

% \shouxu{Should we clarify that after $\frac{B_h}{B_i}$, it is interconnect bandwidth bound?}
\noindent\textbf{Different Sensitivities to Offloading:}
Memory/compute-bound operations react differently to the offloading ratio.
To quantify this, we use the Effective Bandwidth ($\mathcal{EB}$) as the performance metric, which is defined as the total fetched data from memory ($C$)~\footnote{$C$ is the size of model weights or KV cache for linear or attention.} divided by the overall operation execution latency: $\mathcal{EB} = \frac{C}{\max(T_{\text{comp}}, T_{\text{mem}})}$, where $T_{\text{comp}}$ is the computation time and $T_{\text{mem}}$ is the memory access time. \autoref{fig:effective_bandwidth} shows how $\mathcal{EB}$ of both compute- and memory-bound kernels changes as the offloading ratio increases. We use $\mathcal{EB}$ as the primary performance metric because it provides a unified standard to easily identify whether an operation is memory- or compute-bound.
% \jxlin{Should we add once setence explain why use EB as the metric?} \shouxu{Can we say something like using this metric, we can easily identify an operation is memory or compute bound?}

% Memory- and compute-bound operations exhibit distinctly different sensitivities to the memory offloading ratio. To quantify this behavior, we evaluate the Effective Bandwidth ($\mathcal{EB}$). We adopt $\mathcal{EB}$ as our primary performance metric because it provides a unified standard to easily identify the underlying bottleneck of an operation: an $\mathcal{EB}$ approaching the theoretical hardware memory bandwidth limit clearly indicates a memory-bound kernel, whereas a significantly lower $\mathcal{EB}$ reveals a compute-bound kernel. Formally, $\mathcal{EB}$ is defined as the total data fetched from memory ($C$)~\footnote{$C$ represents the size of the model weights or KV cache for linear or attention operations, respectively.} divided by the overall operation execution latency: $\mathcal{EB} = \frac{C}{\max(T_{\text{comp}}, T_{\text{mem}})}$, where $T_{\text{comp}}$ is the computation time and $T_{\text{mem}}$ is the memory access time. \autoref{fig:effective_bandwidth} illustrates how the $\mathcal{EB}$ of both compute- and memory-bound kernels changes as the offloading ratio increases:

\begin{denseenum}
    \item \textbf{Memory-bound operations ($T_{\text{comp}} < T_{\text{mem}}$):} The operation performance is bounded by $T_{\text{mem}}$. Because \name{} concurrently reads from both host and GPU memory, the overall memory latency is bounded by the slower of the two transfers: $T_{\text{mem}} = \max(T_h, T_g)$, where $T_h$ and $T_g$ are the host and GPU data transfer times, respectively. As shown in the figure, $\mathcal{EB}$ initially increases with the offloading ratio, reaches a peak, and then decreases. This peak EB occurs at an optimal ratio of $\frac{B_h}{B_h+B_g}$, where $B_h$ and $B_g$ denote the host interconnect bandwidth and GPU memory bandwidth. At this threshold, the data transfer times are perfectly balanced ($T_h = T_g$). Consequently, the operation fully utilizes both memory bandwidths simultaneously. However, if the offloading ratio is beyond this point, reading from the slower host interconnect ($T_h$) becomes the bottleneck, and performance degrades.

\item \textbf{Compute-bound operations ($T_{\text{comp}} > T_{\text{mem}}$):}
For these operations, $\mathcal{EB}$ remains flat initially, since at lower offloading ratios, the operation latency is strictly dominated by the mathematical computation time ($T_{\text{comp}}$) rather than memory time. Because the computation effectively hides the memory transfer latency, offloading weights to the host does not impact the overall EB. However, this flat trend only holds up to a critical threshold (which we refer to as the ``threshold'' of compute-bound operations). Beyond this point, data movement from host memory becomes the bottleneck (\ie when $T_{\text{mem}} = T_h \ge T_{\text{comp}}$). Once this threshold is crossed, the operation transitions into a memory-bound regime dictated by the host interconnect, and $\mathcal{EB}$ begins to decrease following the same behavior as the memory-bound case.
\end{denseenum}

\noindent\textbf{Suboptimality of Uniform Offloading:}A uniform method that assigns the same ratio $r$ to all operations will be suboptimal. As shown in \autoref{fig:effective_bandwidth}, the suboptimal show up in two regions.
\textbf{Region 1 ($r \le \frac{B_h}{B_h + B_g}$):} the memory-bound operation is starved and actively benefits from increasing its offloading ratio toward its optimal peak, thus the offloading ratio for compute-bound operations should be allocated to memory-bound operations.
\textbf{Region 2 ($\frac{B_h}{B_h + B_g} < r < \frac{B_h}{B_i}$):} memory-bound operations suffer from effective bandwidth loss due to offloading. While compute-bound operations are still insensitive, the offloading budget should be shifted from memory-bound operations to compute-bound operations.

% Consider two operations with identical sizes: one is memory-bound and the other is compute-bound. Our goal is to assign offloading ratios to these two operations such that the overall offloading ratio is $R$.

% If $R \le \frac{B_h}{B_h + B_g}$, a uniform assignment is suboptimal. In this regime, the memory-bound operation benefits from increasing its offloading ratio toward the optimal point $x^* = \frac{B_h}{B_h + B_g}$, where its effective bandwidth is maximized. In contrast, the compute-bound operation is insensitive to offloading as long as it remains compute-bound. Therefore, we can reduce the offloading ratio of the compute-bound operation and reallocate that budget to the memory-bound operation, strictly improving overall performance.

% If $\frac{B_h}{B_h + B_g} < R < \frac{B_h}{B_i}$, a uniform assignment remains suboptimal. In this case, the memory-bound operation suffers from excessive offloading, which reduces its effective bandwidth. By decreasing its offloading ratio toward $x^*$, we can improve its performance. The freed offloading budget can be reassigned to the compute-bound operation, which can absorb additional offloading without performance degradation until it becomes memory-bound. This reallocation again leads to a strictly better solution than the uniform assignment.

\subsubsection{Optimal Greedy Offload}
Because varying operations yield different performance curves, we formulate this ratio offloading allocation as an optimization problem.

\noindent\textbf{Problem Formulation:}
Say the global offload ratio is $OR$. Given a set of operations $\mathcal{F}$, we assign an offloading ratio $x_i \in [0,1]$ to each operation $F_i$. The sum of offloaded data from all operations should match the global ratio constraint.

Our objective is to minimize the end-to-end execution time, which is the sum of individual operation latencies. Using the effective bandwidth model, the latency of operation $F_i$ can be written as $\frac{C_i}{\mathcal{EB}(x_i)}$ under offloading ratio $x_i$.

% where $R \in [0,1]$ is the overall offloading ratio.

\noindent\textbf{Greedy Offload Algorithm:}
We prove that the optimal allocation strictly follows a multi-stage greedy algorithm based on each operator's turning point shown in \autoref{fig:effective_bandwidth}. We allocate the offloading ratio to operations in three distinct phases:

\begin{denseenum}
    \item \textbf{Allocate to Memory-Bound Operations.} It first allocates the offloading budget all to memory-bound operations, since it directly improves $\mathcal{EB}$. The exact distribution of the budget between the memory-bound operations does not matter if no memory operation's ratio is beyond its turning point.
    \item \textbf{Saturate Compute-Bound Operations.} Once all memory-bound operations reach their peak $\mathcal{EB}$, the remaining offloading budget should be distributed to the compute-bound operations. Similar to the first phase, the exact distribution among compute-bound operations does not affect optimality, as long as no compute-bound operation's ratio is beyond its turning point.
    \item \textbf{Arbitrary Allocation.} If the global offloading ratio is large enough that all operations (both memory- and compute-bound) have reached their turning points, the remaining budget can be assigned arbitrarily across any operations without affecting the end-to-end latency.
\end{denseenum}

% \noindent\textbf{Proof of Optimality:}
% Here, we sketch the high-level intuition that proofs the greedy algroithm is optimal, and the detailed formal proof is provided in \autoref{sec:appendix_greedy_proof}. We prove via contradiction that assigning any budget to compute-bound operations before all memory-bound operations reach their turning points results in strictly sub-optimal end-to-end execution time. Once all optimal turning points are met, the latency objective function reduces to a constant under the budget constraint, rendering the exact intra-class distribution of the budget irrelevant to overall optimality. \shouxu{can be shortened for space?}

\noindent\textbf{Proof of Optimality:} We prove via contradiction (\autoref{sec:appendix_greedy_proof}) that our greedy algorithm is optimal.

\begin{figure}[t!]
    \centering
    \begin{subfigure}[t]{0.245\textwidth}
        \centering
        \includegraphics[clip,width=\textwidth]{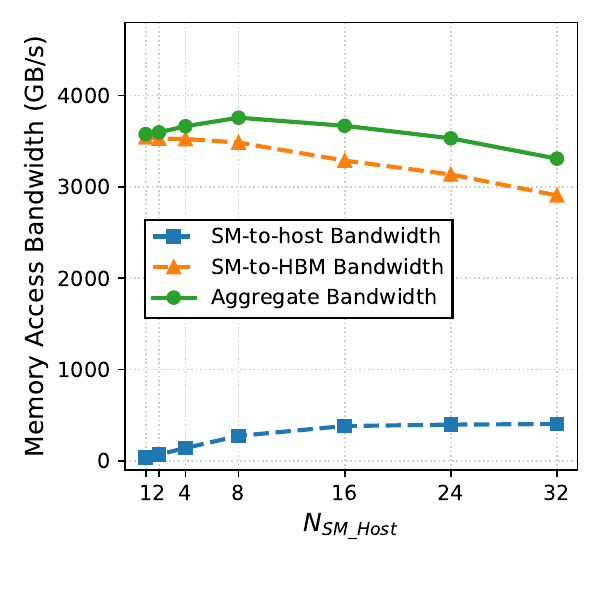}
% 353949
\vspace{-0.3cm}
\scaption{ $N\_{inflight}=3$ , $N_{SM\_HBM}=100$ and we vary $N_{SM\_Host}$.}
        \label{subfig:congestion_SM}
    \end{subfigure}%
~
    \begin{subfigure}[t]{0.245\textwidth}
        \centering
        \includegraphics[clip,width=\textwidth]{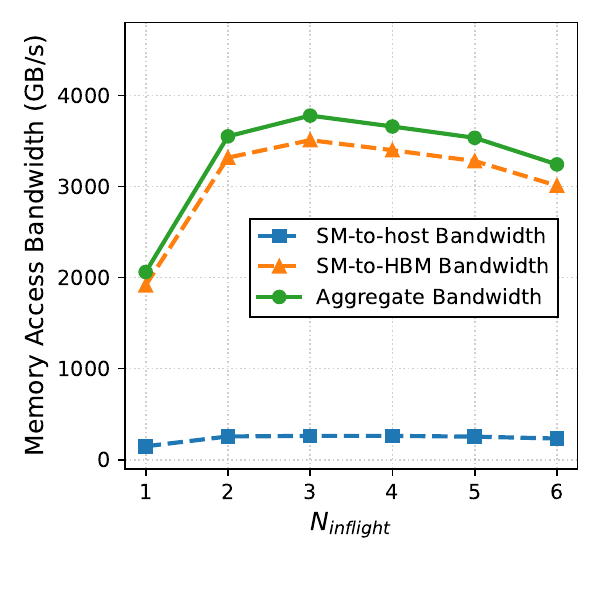}
        \vspace{-0.3cm}
        \scaption{$N_{SM\_HOST}=8$, $N_{SM\_HBM}=100$ and we vary $N\_{inflight}$.}
        \label{subfig:congestion_inflight}
    \end{subfigure}%
~
    \vspace{-0.3cm}
    \scaption{Local/remote memory access congestion on GH200. Increased $N_{SM\_Host}$ or $N_{inflight}$ reduces SM-to-HBM bandwidth.}
    \vspace{-0.2cm}
    \label{fig:congestion} 
\end{figure}

\subsection{Efficient Host Memory Access with TMA}
\label{subsec:congestion_control}

In \name{}, SMs utilize the TMA to concurrently access both host and GPU memory. To our knowledge, \textit{\name{} is the first system to leverage TMA for host memory accesses}. However, relying on TMA for cross-interconnect data movement requires careful orchestration to achieve optimal performance and avoid severe resource contention.

\subsubsection{Congestion Control over Local/Remote Access}
We observe that issuing an unconstrained number of inflight TMA requests to host memory triggers interconnect congestion, which degrades the GPU's local HBM performance. 

The total volume of SM-to-host inflight TMA memory requests can be estimated as $N_{SM\_host} \times N_{inflight}$, where $N_{SM\_host}$ is the number of SMs reading from host memory, and $N_{inflight}$ is the number of inflight TMA requests issued per SM. \autoref{fig:congestion} illustrates how unconstrained $N_{SM}$ or $N_{inflight}$ causes contention that degrades local HBM throughput on the Grace Hopper Superchip. 

In the first experiment (\autoref{subfig:congestion_SM}), we dedicate 100 SMs to read from local GPU HBM while varying $N_{SM}$ reading from remote host memory via NVLink-C2C. As shown in the figure, when $N_{SM}$ exceeds 8, the SM-to-HBM bandwidth drops significantly, dragging down the system's total aggregate bandwidth.
In the second experiment (\autoref{subfig:congestion_inflight}), we fixes the number of SMs reading from both host and GPU, but varies $N_{inflight}$ issued per SM. The results show that excessive in-flight requests from individual SMs to the host also cause a severe slowdown in aggregate bandwidth.

\noindent\textbf{Root Cause Analysis:} 
We observe that this bandwidth degradation does not occur when SMs only access HBM, indicating there is contention between the local and remote data paths. While proprietary GPU interconnect details are undisclosed, we hypothesize this stems from poor resource isolation within the memory hierarchy. Specifically, once the host-GPU interconnect saturates, excess in-flight remote requests cannot drain. These stalled requests accumulate and exhaust shared internal resources (such as L2 cache MSHRs or interconnect routing buffers) which starves the local SM-to-HBM traffic and severely degrades overall performance.

% latency asymmetry. Remote host memory accesses incur significantly higher latency than local HBM accesses. A massive volume of long-latency remote requests likely exhausts shared interconnect resources (e.g., network-on-chip credits). \jxlin{Actually, I think it might be a backpressure problem, since the congestion point starts when the host memory starts to saturate.} This latency differential creates an inherent unfairness: pending host requests disproportionately occupy the interconnect, starving the much faster local HBM accesses of necessary routing credits and stalling the entire memory hierarchy.

\noindent\textbf{Congestion Control Strategy:} 
\name{} applies congestion control to remote memory accesses by limiting the ($N_{inflight}$) and ($N_{SM}$). 
% Together, these mechanisms prevent remote memory traffic from overwhelming the intra-GPU interconnect and congesting the HBM data path.
To limit $N_{inflight}$, \name{} uses the concept of a \textit{congestion window} from traditional computer networking. Each SM fetching from the host is assigned a maximum allowance of in-flight requests that is sized to saturate the SM-to-host memory bandwidth without exceeding it. A given SM's TMA engine will strictly cap its asynchronous memory requests to this window. Unlike traditional networking protocols that dynamically adjust this window at runtime (which would incur significant runtime overhead for GPU kernels), \name{} assigns the congestion window statically. This optimal static window size is calculated via a lightweight parameter-sweeping profiler executed prior to kernel launch.

\name{} limits $N_{SM}$ by ensuring we do not over-assign SMs beyond what is necessary. The optimal host $N_{SM}$ is determined by two factors: 1) interconnect congestion, and 2) how many SMs are needed to process the offloaded data on the host. To find the optimal $N_{SM}$, \name{} performs an offline parameter sweep—fixing the required local GPU SMs while varying the host SMs—to identify the exact SM allocation to the host that maximizes end-to-end throughput. This ensures the system provisions exactly enough SMs to saturate the computation needs and avoid congestion.

\begin{table}[t!]
\centering
\footnotesize
\begin{tabular}{@{}rrr@{}}
\hline
\textbf{Size ($N$)} & \textbf{Host-GPU Traffic} & \textbf{Read Amplification} \\
\hline
256 & 102.76 MB & 1.05$\times$ \\
512 & 205.52 MB & 2.10$\times$ \\
1024 & 411.04 MB & 4.19$\times$ \\
2048 & 822.08 MB & 8.39$\times$ \\
4096 & 1.64 GB & 16.78$\times$ \\
\hline
\end{tabular}
\scaption{Host-GPU memory traffic across varying sizes ($N$). The size of the matrix offloaded to the host is 98 MB.}
\vspace{-0.4cm}
\label{tab:read_amplification}
\end{table}
\subsubsection{TMA Multicast for Uncacheable Host Access}
A critical challenge in direct-access offloading is that host-homed memory is uncacheable by the GPU, completely bypassing the L2 cache. Our experiment shows that this architectural limitation applies to both traditional PCIe-connected systems and even coherent CPU-GPU platforms~\footnote{Even on the Grace Hopper Superchip, while the NVLink-C2C supports coherent unified memory, direct host-homed memory accesses bypass the GPU's L2 cache to maintain hardware coherence without triggering a page migration.}. 

Consequently, direct host access suffers from severe \textit{read amplification}: if multiple SMs request the same remote tile, the identical data chunk must be transmitted across the host-GPU interconnect multiple times as it cannot cached by GPU L2. For instance, in the scenario depicted in \autoref{subfig:weights}, output tiles $t_0$ and $t_1$ both require tile row 1 of matrix $A$, that row is fetched twice over the long-latency interconnect by SM0 and SM1. As shown in \autoref{tab:read_amplification}, we measured the read amplification effect with different size $N$ in \autoref{subfig:weights} and the total host-to-GPU memory transfer scales linearly, reaching 16x more data transfer over CPU-GPU interconnect at $N = 4096$.

% \name{} leverages TMA Multicast capabilities to eliminate redundant interconnect traffic. Instead of each SM issuing independent remote fetches, a single warp within a Thread Block Cluster issues a TMA multicast instruction specifying a destination SM mask. The TMA engine fetches the host tile across the host-GPU interconnect exactly once. Upon arriving at the GPU's L2 cache subsystem, the hardware leverages the high-speed intra-cluster Network-on-Chip (NoC) to spatially broadcast the payload using Distributed Shared Memory (DSMEM). Physically, the NoC routes the single data payload and simultaneously writes it directly into the SRAM banks of the shared memory (SMEM) for all targeted SMs, bypassing the L1 caches. To ensure data consistency, the TMA engine signals a hardware-backed asynchronous transaction barrier (\texttt{mbarrier}) in each receiving SM, allowing all threads to safely synchronize and overlap active computation while the multicast completes.
To eliminate this interconnect bottleneck, \name{} uses TMA Multicast and host-aware tile scheduling:
\begin{denseitemize}
    \item \textbf{TMA Multicast:} \name{} leverages the advanced TMA Multicast capabilities~\cite{nvidia_cuda_programming_guide_tma}. Instead of each SM issuing independent remote fetches, TMA fetches the host tile across the host-GPU interconnect exactly once. Upon host data arriving at the GPU, the GPU leverages the high-speed intra-cluster Network-on-Chip (NoC) to spatially broadcast the data using Distributed Shared Memory (DSMEM). Physically, the NoC routes the single data payload and simultaneously writes it directly into the SRAM banks of the SMEM for all targeted SMs.
    \item \textbf{Host-Locality-First Tile Scheduling:} 
    Unlike traditional scheduling that optimizes for HBM locality or L2 cache hit rates, \name{} implements \textit{host-locality-first} SM scheduling to enable efficient multicast. Specifically, the producer warps of SMs accessing the same host tile chunk are grouped into the same thread block cluster, which are guaranteed to be co-scheduled simultaneously. Only the producer warp of one SM in the cluster initiates multicast, which reads the tile chunk from host memory once and then stores it to the SMEM slot of all SMs in the cluster.

    % TMA Multicast strictly requires the producer warps of SMs accessing the same host tile chunk to be scheduled simultaneously in the thread block cluster. To satisfy this, \name{} implements \textit{host-locality-first} tile scheduling. Unlike traditional tile schedulers that optimize for HBM locality or L2 cache hit rates, our tile scheduler guarantees that all output tiles requiring the same host-memory chunk are grouped into the same cluster to be co-scheduled simultaneously. This ensures that the TMA multicast mechanism can successfully serve all consumer warps using only a single host-to-GPU memory transaction.
\end{denseitemize}

\section{Implementation}
\label{subsec:impl}

We implemented \name{} in approximately 1,300 lines of Python and 2,700 lines of CUDA C++, leveraging CUTLASS to construct our direct-access kernels. The code is open-sourced at \url{https://github.com/shouxulin/DirectAccessKernel.git}. Our system is designed for seamless integration into existing deep learning ecosystems and easy portability, with four key aspects:

\noindent\textbf{PyTorch Interface:} We expose two custom PyTorch modules, \textit{SplitK\_GEMM} and \textit{SplitK\_FlashAttn}, which serve as drop-in replacements for the native \texttt{nn.Linear} and attention layers. This API design allows users to adopt \name{} by simply swapping the corresponding layers during initialization, without modifying the model architecture code.

\noindent\textbf{Supporting FlashAttention:} \textit{SplitK\_FlashAttn} follows the same interface as Pytorch's Scaled Dot-Product Attention (SDPA)~\cite{pytorch_docs_sdpa}, while extending it with support for partitioned KV cache access via TMA. It partitions the KV cache along the batch dimension—storing the cache for a subset of requests in local GPU memory and the remainder in remote host memory. Similar to our linear operation kernels, \textit{SplitK\_FlashAttn} supports placing the KV cache into either host memory or HBM based on a given offload ratio.

% \textit{SplitK\_FlashAttn} implements the Scaled Dot-Product Attention (SDPA) algorithm~\cite{pytorch_docs_sdpa}, extending it with support for asynchronous KV cache loading via the TMA. 

\noindent\textbf{CUDA Graph Optimization:} To eliminate CPU-side kernel launch overheads during the latency-sensitive decoding phase, \name{} fully supports CUDA Graphs. We pre-allocate the KV cache using PyTorch's \texttt{StaticCache}, to capture the execution graph using standard APIs, and seamlessly replay it across all subsequent decoding steps.

\noindent\textbf{Cross-Architecture Adaptation:}  To adapt to different GPU generations, we swap the underlying CUTE compute atoms to perform matrix multiplication for that specific architecture (\eg adapting to different tensor core generations). Furthermore, we expose tunable parameters for the TMA chunk size and maximum in-flight requests. This allows the system to easily calibrate to different SMEM size and interconnect bandwidth limits on different architectures.

\begin{figure*}[h]
\begin{subfigure}[b]{0.49\linewidth}
        \centering
    \includegraphics[width=0.99\linewidth]{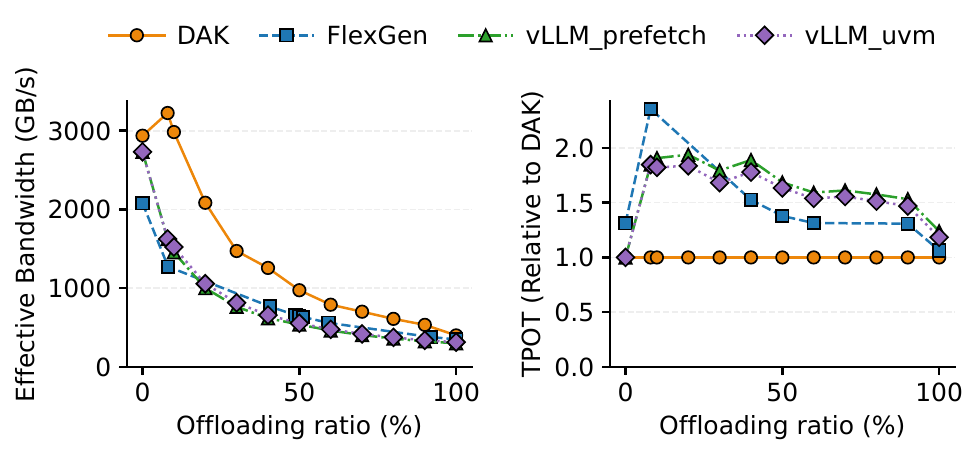}
    \scaption{OPT-30b on GH200}
    \label{subfig:of_1}
\end{subfigure}
~
\begin{subfigure}[b]{0.49\linewidth}
        \centering
    \includegraphics[width=0.9\linewidth]{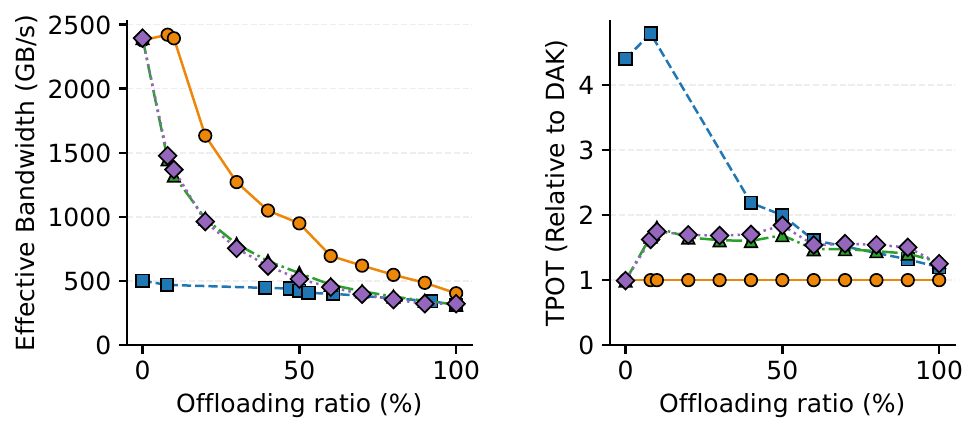}
    \scaption{OPT-6.7b on GH200}
    \label{subfig:of_2}
\end{subfigure}
\hfill

\begin{subfigure}[b]{0.49\linewidth}
        \centering
    % Scaled to match figure (b)
    \includegraphics[width=0.9\linewidth]{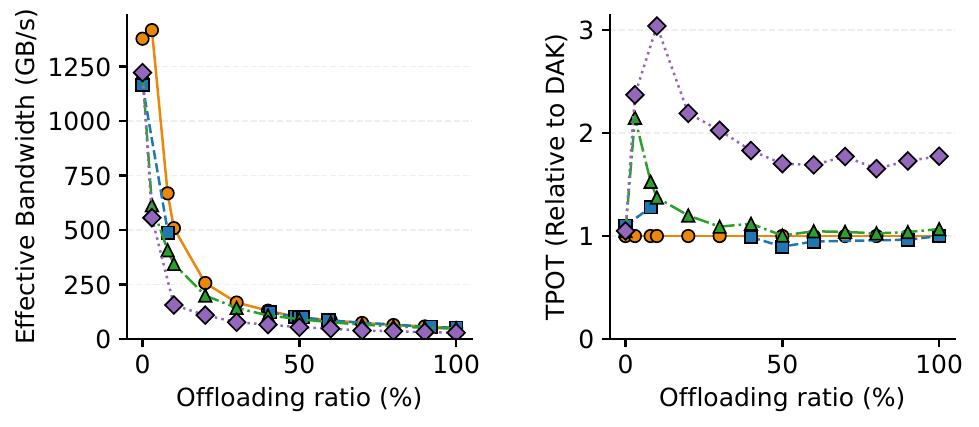}
    \scaption{OPT-30b on RTX6000 Pro Blackwell}
    \label{subfig:of_3}
\end{subfigure}
~
\begin{subfigure}[b]{0.49\linewidth}
        \centering
    % Scaled to match figure (b)
    \includegraphics[width=0.9\linewidth]{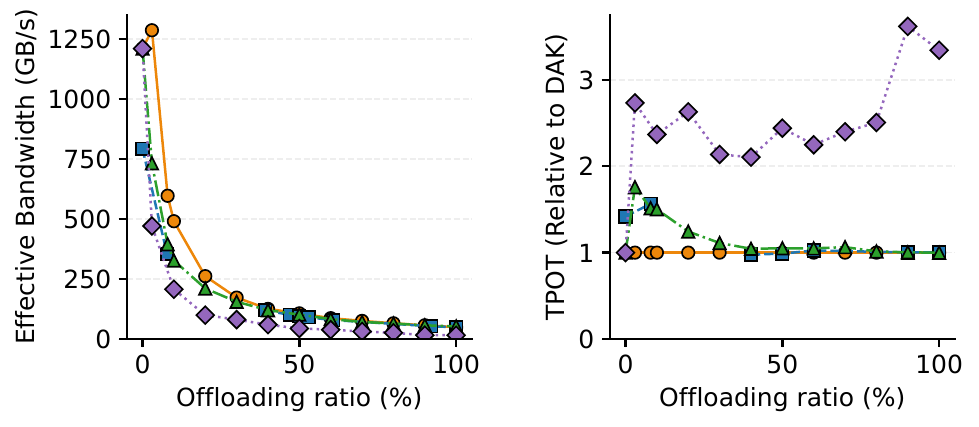}
    \scaption{OPT-6.7b and on RTX6000 Pro Blackwell}
    \label{subfig:of_4}
\end{subfigure}
\scaption{Performance w/ batch size = 8 under vary offloading ratios.}
\vspace{-0.4cm}
\label{fig:offload_ratio}
\end{figure*}
\section{Evaluation}
\label{sec:eval}
We evaluate \name{} by answering the following questions:
(1) How does \name{} compare to baselines across various offloading ratios, hardware architectures, and models?
(2) Can \name{} enable large models to run efficiently when their memory footprints exceed GPU capacity (\autoref{eval:e2e})?
(3) What is the performance gain of the greedy offloading algorithm?
(4) How effective are the TMA optimizations (\autoref{eval:opt})?

\noindent\textbf{Testbed \& Workloads:} 
We evaluate our system across two distinct NVIDIA GPU architectures to capture a range of interconnect technologies and GPU generations:
\begin{denseitemize}
\item \textbf{Grace Hopper (GH200):} An NVLink-C2C-based system featuring a 900 GB/s host-GPU interconnect and 480 GB of host memory. The GPU includes 96 GB of local HBM3 with 4.0 TB/s of bandwidth.
\item \textbf{RTX 6000 Pro Blackwell:} A PCIe Gen5-based system with a unidirectional host-GPU bandwidth of 64 GB/s and 512 GB of host memory. The GPU features 96 GB of local GDDR7 with 1.8 TB/s of bandwidth.
\end{denseitemize}
Our experiments are conducted under an offline, batched-inference setup. We evaluate a diverse set of LLMs across varying batch sizes and prompt lengths, decoding 32 tokens per request. Unless otherwise specified, the prompt length is set to 32. We test across multiple models (OPT-30b, OPT-6.7b, and Llama-2-7b). In our evaluation, we compare our system against three baselines: \textit{FlexGen}~\cite{sheng2023flexgen}, a state-of-the-art offloading framework that uses double-buffered, layer-by-layer prefetching to manage weights and KV caches across tiered memory; \textit{vLLM-prefetch}~\cite{dao2022flashattention}, a config of vLLM that implements asynchronous prefetch of KV blocks and weights from host memory to GPU HBM; and \textit{vLLM-uvm}, a vLLM variant that utilizes NVIDIA Unified Virtual Memory (UVM), which relies on hardware page faults and system managed paging migrations to access host memory. 

% \jxlin{@Shouxu, could you check and add more detail to this setup if necessary?}

\subsection{\name{} End-to-End Performance}
\label{eval:e2e}
In this experiment, we compare \name{}'s performance against baseline systems by sweeping the offloading ratio from 0\% to 100\% for the OPT-30B and OPT-6.7B models.
We evaluate using two primary metrics: (1) \textbf{Time Per Output Token (TPOT)}, representing the end-to-end decoding latency for a single token, and (2) \textbf{Effective Bandwidth ($\mathcal{EB}$)} defined as the total model size divided by TPOT, which shows the achieved aggregated memory bandwidth.

\noindent\textbf{Weights Offload:}
We first measure performance using a batch size of 8.
In this experiment, because the batch size is small, the memory footprint of the KV cache is also small, meaning the offloaded memory is primarily model weights. With a small batch size, the decoding is mostly memory-bound.

As shown in \autoref{fig:offload_ratio}, \name{} consistently achieves the highest $\mathcal{EB}$ across all models and offload ratios. On the GH200 system, \name{}'s direct-access approach provides massive advantages across all offload ratios, yielding 1.5$\times$ to 5$\times$ performance gains. \name{} also achieves optimal aggregate system memory bandwidth through offloading. For instance, with OPT-30B at a 10\% offload ratio, \name{} sustains 3,300 GB/s, achieving a near-optimal aggregate memory bandwidth that approaches the theoretical peak of NVLink-C2C + HBM. In contrast, prefetching baselines suffer bandwidth degradation across all offload ratios. \name{}'s high memory bandwidth also translates directly into a lower TPOT. As a note, in \autoref{subfig:of_2}, FlexGen exhibits low performance on OPT-6.7B because kernel launch overhead dominates for smaller models, whereas vLLM and \name{} utilize CUDA Graphs to mitigate this overhead entirely.
% \jxlin{@Shouxu, could you please help double check this claim?} \shouxu{Yeah, its correct.}

As shown in \autoref{subfig:of_3} and \autoref{subfig:of_4}, on the PCIe-based RTX 6000 Blackwell system, \name{} remains strictly better to all prefetching schemes. It achieves near-ideal aggregate system memory bandwidth, demonstrating a ~4\% improvement through offloading that closely approaches the theoretical limits of the PCIe Gen5 interconnect. At low-to-medium offload ratios (0\%--40\%), \name{} delivers a 1.3$\times$ to 3$\times$ improvement in both $\mathcal{EB}$ and TPOT compared to vLLM-prefetch and FlexGen. Beyond a 40\% offload ratio, the performance of \name{}, vLLM-prefetch, and FlexGen converges, as operation latency becomes entirely bottlenecked by the physical limits of transferring weights over the low-bandwidth PCIe link. The vLLM-UVM baseline performs significantly worse than all other systems due to the severe latency penalties of page-fault and paging migration overheads over PCIe.

% \begin{figure}[t!]
%     \centering
%     \includegraphics[width=0.9\linewidth]{figures/eval/GH200_opt-30b_prompt_len32_bsz512_all_metrics.pdf}
%     \vspace{-0.2cm}
%     \scaption{OPT-30b on GH200 with batch size(bsz)=512.}
%      \vspace{-0.2cm}
%     \label{fig:kvcacheoffload}
% \end{figure}

\begin{figure}[t!]
\begin{subfigure}[b]{0.48\linewidth}
        \centering
    \includegraphics[width=0.99\linewidth]{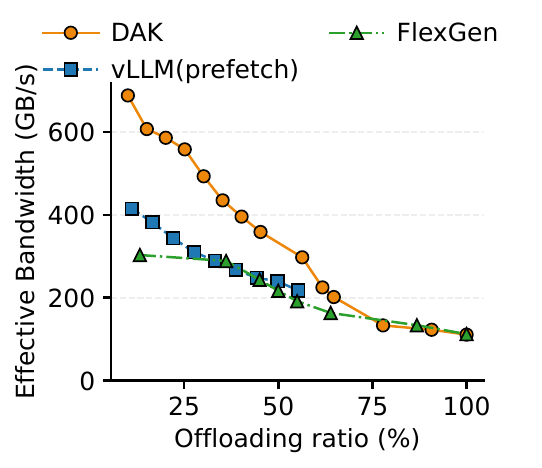}
    \scaption{OPT-30b on GH200}
    \label{subfig:of_1}
\end{subfigure}
~
\begin{subfigure}[b]{0.48\linewidth}
        \centering
    \includegraphics[width=0.99\linewidth]{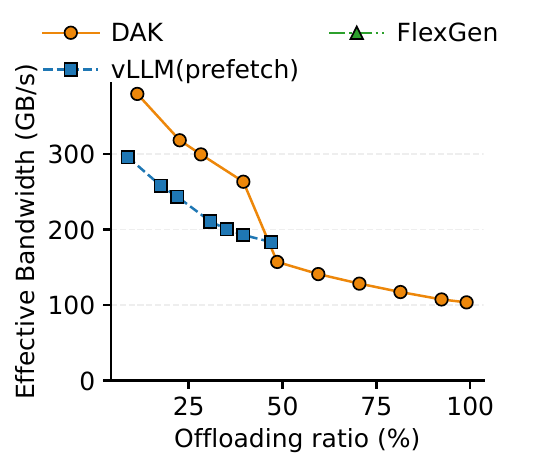}
    \scaption{Llama-2-7b on GH200}
    \label{subfig:of_2}
\end{subfigure}

\scaption{Performance w/ batch size = 512 under vary offloading ratios}
\vspace{-0.3cm}
\label{fig:kvcacheoffload}
\end{figure}

\noindent\textbf{KV Cache Offload:}
We next evaluate performance using a larger batch size of 512. At this scale, the KV cache footprint increases to 45 GB, requiring the system to concurrently manage the offloading of both weights and KV caches. Furthermore, the decoding phase becomes a heterogeneous mix of compute-bound operations (linear layers) and memory-bound operations (attention).
\autoref{fig:kvcacheoffload} shows that \name{} improves up to 2.1x performance compared to baselines on the GH200 system under OPT, and Llama model. \name{} achieves this performance for two key reasons: (1) its direct-access paradigm fully utilizes aggregate system bandwidth to accelerate memory-bound attention operations, (2) its offloading algorithm analytically determines the optimal, operation-specific offloading ratios for both weights and KV caches. In contrast, baseline systems apply a naive, uniform offloading strategy that ignores these distinct operation characteristics, and (3) \name{} can efficiently handle the read amplification problem with TMA-multicast under large batch sizes.
Note that vLLM data points are omitted for offload ratios exceeding 50\%, as this is the threshold where vLLM begins offloading the KV cache to CPU memory. In these cases, vLLM falls back to splitting the workload into multiple sub-batches and interleaving the prefill and decode phases, which prevents a clean, isolated measurement of the decode latency for the entire batch. Additionally, FlexGen results for Llama are missing because it is not supported by FlexGen.

\begin{figure}[t!]
    \centering
    \includegraphics[width=0.98\linewidth]{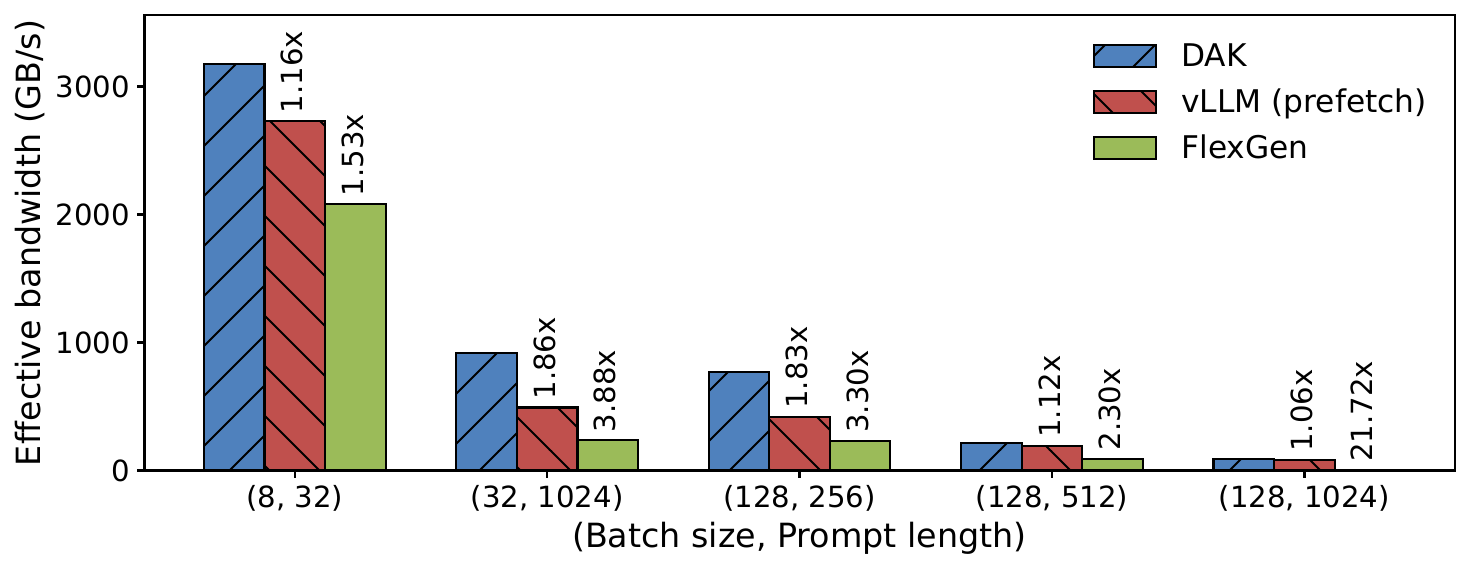}
    \footnotesize
     \begin{tabular}{ccccc}
    \hline
    bsz & prompt\_len & model  & KV cache  & global offload ratio \\ \hline
    8   & 32         & 55.6 GB      & 0.7    GB    & 0\%                \\ 
    32  & 1024       & 55.6 GB      & 46.51  GB    & 6\%                 \\ 
    128 & 256        & 55.6 GB      & 50.73  GB    & 10\%                 \\ 
    128 & 512        & 55.6 GB      & 95.83  GB    & 37\%                 \\ 
    128 & 1024       & 55.6 GB      & 186.03 GB    & 60\%                 \\ \hline
    \end{tabular}
    
    \scaption{Figure shows the performance for GH200 under different configurations. The table shows the memory footprint for different configurations and the corresponding global memory offload ratio.}
     \vspace{-0.2cm}
    \label{fig:bw-config}
\end{figure}

% \zy{not end to end, but more like "optimial", give na setup what's the best performance we can get }
\noindent\textbf{Optimal Model Offloading:}
Now, rather than manually sweeping offload ratios, we evaluate \name{}'s capability to efficiently run large models on GPUs with limited memory capacity, with the global offloading ratio determined by the real available GPU memory. As shown in \autoref{fig:bw-config}, we vary the batch size and prompt length for OPT-30B. Large batch sizes and extended prompt lengths increase the total memory footprint well beyond the available GPU HBM, which dynamically dictates the required global offload ratio.

We compare the maximum achievable performance across these configurations against vLLM and FlexGen. \name{} consistently outperforms all baselines across every configuration, achieving 1.06$\times$ to 1.83$\times$ speedups over vLLM, and 1.53$\times$ to 21$\times$ speedups over FlexGen. These significant gains highlight the combined effectiveness of \name{}'s optimal offloading algorithm and highly efficient direct-access. In \autoref{appendix:addi}, we show additional results on OPT-6.7B.

\begin{figure}[t]
\begin{subfigure}[b]{0.48\linewidth}
        \centering
    \includegraphics[width=0.99\linewidth]{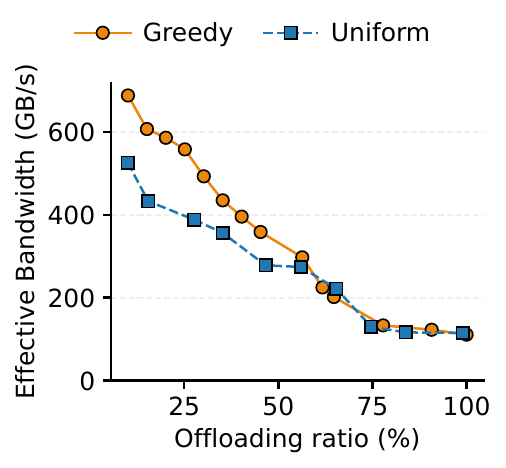}
    \scaption{OPT-30b on GH200}
    \label{subfig:of_1}
\end{subfigure}
~
\begin{subfigure}[b]{0.48\linewidth}
        \centering
    \includegraphics[width=0.99\linewidth]{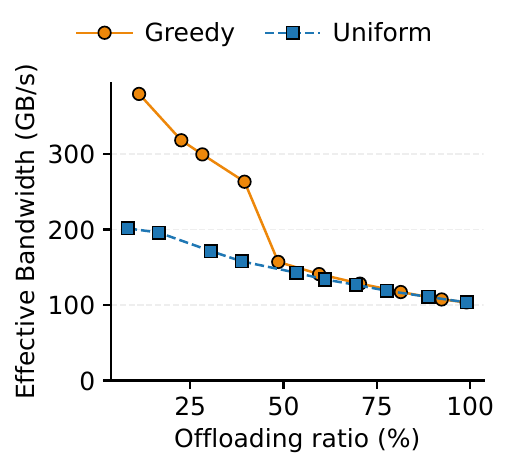}
    \scaption{Llama-2-7b on GH200}
    \label{subfig:of_2}
\end{subfigure}

\scaption{Compare greedy with uniform offloading. Batch size is 512.}
\vspace{-0.3cm}
\label{fig:greedy}
\end{figure}
\begin{figure}[t!]
    \centering
    \begin{subfigure}[t]{0.24\textwidth}
        \centering
        \includegraphics[clip,width=\textwidth]{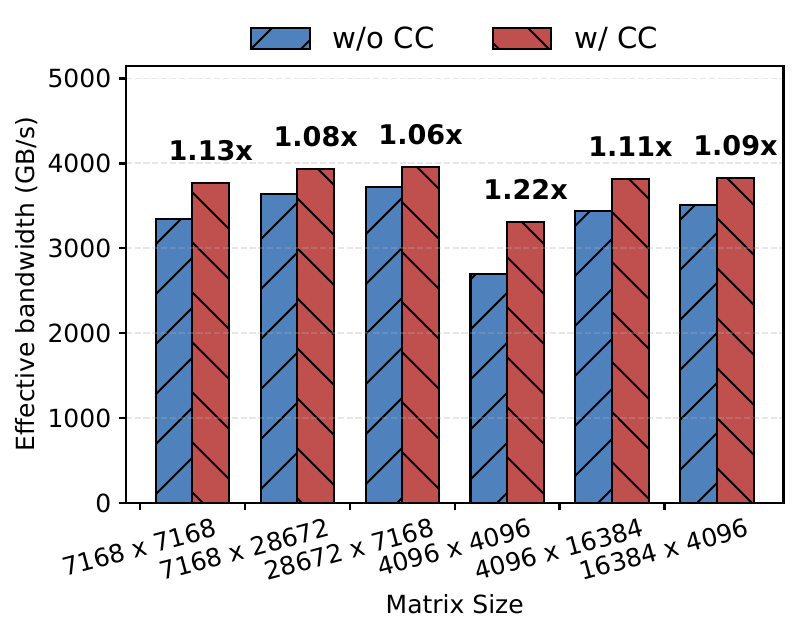}
% 353949
\scaption{Congestion Control Effect}
        \label{subfig:micro_cc}
    \end{subfigure}%
~
    \begin{subfigure}[t]{0.24\textwidth}
        \centering
        \includegraphics[clip,width=\textwidth]{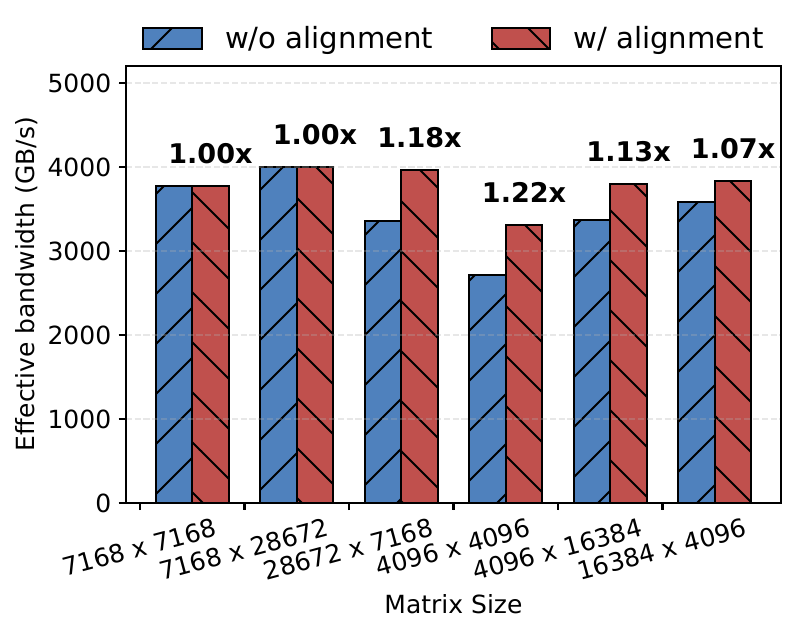}
        \scaption{Kernel Alignment Effect}
        \label{subfig:micro_alignment}
    \end{subfigure}%
~
    \vspace{-0.2cm}
    \scaption{Effectiveness of congestion control and kernel alignment on different matrix sizes.}
    \vspace{-0.3cm}
    \label{fig:micro} 
\end{figure}

\subsection{Effectiveness of \name{} Optimizations}
\label{eval:opt}
Next, we run benchmarks on the GH200 system and evaluate \name{}'s key design decisions.

\noindent\textbf{Greedy Offload Performance:}
We show that our greedy offload algorithm successfully incorporates the compute- and memory-bound characteristics of individual operations to select the optimal offloading ratio for each. We compare the greedy strategy against a uniform offloading baseline using a batch size of 512, which is a workload that has a mix of compute- and memory-bound operations.
As shown in \autoref{fig:greedy}, the greedy algorithm outperforms uniform offloading by 1.5$\times$ when the offloading ratio is below 60\%. This is because the greedy approach identifies the turning point of each kernel's $\mathcal{EB}$ and prioritizes offloading compute-bound kernels, which are significantly less sensitive to remote memory access latencies. When the offloading ratio exceeds 60\%, the uniform and greedy algorithms converge in performance. At such high ratios, all operations become bottlenecked by the interconnect bandwidth; consequently, assigning the offloading budget to any specific kernel yields the same overall performance. This result aligns perfectly with our theoretical proof in \autoref{subsec:offloading_algo}.

\noindent\textbf{Congestion Control Efficacy:} 
\autoref{subfig:micro_cc} evaluates the GEMM performance across various matrix sizes, with and without our congestion control mechanism. By carefully regulating in-flight TMA credits, our algorithm successfully prevents unconstrained memory requests from overwhelming the intra-GPU interconnect. Consequently, across all evaluated matrix sizes, this controlled dispatch mitigates interconnect congestion and improves performance by up to 1.22$\times$.

\begin{figure}[t!]
    \centering
    \includegraphics[width=\linewidth]{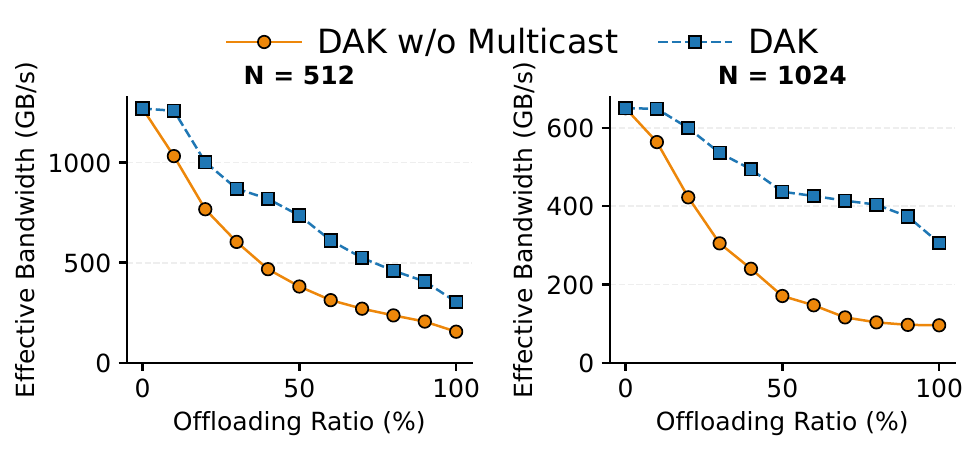}
    \vspace{-0.3cm}
    \scaption{GEMM of weights (7168,7168) and hidden states (7168, $N$).}
     \vspace{-0.4cm}
    \label{fig:multicast}
\end{figure}
\noindent\textbf{TMA Multicast Efficiency:} 
\autoref{fig:multicast} shows the performance benefits of TMA multicast. We measure the performance of GEMM over a $(7168, 7168)$ weight matrix with a $(7168, N)$ hidden state matrix. In this setup, we scale the batch dimension $N$ (which is determined by the LLM inference batch size) from 512 to 1024. At $N=512$, enabling TMA multicast improves performance by 1.3$\times$ compared to the non-multicast baseline. Crucially, this performance gap widens significantly to 2.5$\times$ as $N$ increases to 1024. This scaling behavior perfectly aligns with our earlier analysis (\autoref{tab:read_amplification}): as $N$ grows, independent requests to uncacheable host memory suffer from severe read amplification. By broadcasting a single fetched tile to multiple SMs on-chip, TMA multicast eliminates these redundant host-to-GPU transfers and improves performance.

\noindent\textbf{Kernel Alignment Efficacy:}
As shown in \autoref{subfig:micro_alignment}, \name{} ensures that the number of assigned tiles is evenly divisible across the allocated SMs to prevent tail latencies caused by partial or unbalanced execution waves. \autoref{subfig:micro_alignment} shows that this optimization improves throughput by up to 1.2$\times$.
% \input{figures/tex/eval/offloading/gh200_opt-30b_prompt_len_32_bs_8_decoding}

% \input{figures/tex/eval/offloading/gh200_opt-6.7b_prompt_len_32_bs_8_decoding}

% \input{figures/tex/eval/offloading/gh200_opt-30b-prompt_pen_32_bsz_512_decoding}

% \input{figures/tex/eval/offloading/rtx6000_opt-30b_prompt_len_32_bs_8_decoding}

% \input{figures/tex/eval/offloading/rtx6000_opt-6.7b_prompt_len_32_bs_8_decoding}

% \input{figures/tex/eval/ablation_study/gh200_opt-30b_prompt_len_32_bs_8_decoding}
% \input{figures/tex/eval/ablation_study/gh200_opt-6.7b_prompt_len_32_bs_8_decoding}

% \input{figures/tex/eval/eval3/alignment}

% \input{figures/tex/eval/eval3/congestion_control}

% \section{Limitations and Future Work}
% \label{sec:discuss}

\section{Related Work}
\label{sec:related}

\noindent\textbf{Memory Offloading for LLM Inference.}
Existing memory offloading systems ~\cite{sheng2023flexgen, aminabadi2022deepspeed, dao2022flashattention, ren2021zero, li2025fenghuang, huang2020swapadvisor, cao2025moe} rely on a \textit{copy-based prefetching}, which, as we have shown, underperforms compared to a direct-access approach. Prefetching is only advantageous if a compute-bound kernel's bandwidth demand is lower than the interconnect capacity, leaving idle bandwidth to prestage data for \textit{subsequent memory-bound operations}. However, this scenario is rare in high-throughput LLM decoding, where kernels typically saturate the interconnect or exhibit uniform compute-intensity across stages, leaving no idle bandwidth to hide later stage fetch latencies.

% Theoretically, prefetching is only advantageous if the currently executing kernel is extremely compute-bound, meaning its active memory bandwidth requirement is even lower than the CPU-GPU interconnect bandwidth. And the prefetching system will use the leftover, unutilized interconnect bandwidth to asynchronously stage data for the \textit{next memory-bound operation}. However, we observe that this is rare in high-throughput LLM serving pipelines, especially in decoding. Most LLM operations demand memory bandwidth that exceeds the limits of PCIe or NVLink-C2C, or is purely extreme compute-bound in all stages without a compute-memory bound mix.

\noindent\textbf{Direct Access with Managed Memory.}
Previous GPU direct host access often relies on Unified Virtual Memory (UVM)~\cite{li2015evaluation, chien2019performance}, which provides on-demand access via 4KB hardware page faults. However, for large-scale tensor operations, demand-paging incurs severe fault-handling latency and serialization overheads. Even on recent tightly-coupled architectures like GH200 Superchip, evaluations~\cite{schieffer2024harnessing, fusco2024understanding} show that managed memory imposes significant migration and fault-handling bottlenecks for memory-intensive workloads. \name{} uses TMA to avoid paging overhead.

\noindent\textbf{Tensor Memory Accelerator (TMA) and Async Kernels.}
Libraries like CUTLASS~3~\cite{cutlass3}, ThunderKittens~\cite{spector2024thunderkittens}, Mirage~\cite{cheng2025mirage}, and FlashAttention-3~\cite{shah2024flashattention} use TMA for warp-specialized pipelining. However, prior work restricts TMA to \textit{intra-GPU} (HBM-to-SMEM) movement. \name{} repurpose TMA for \textit{inter-tier} data movement, orchestrating direct, asynchronous accesses over the host-GPU interconnect while actively mitigating remote congestion and read amplification.
\section{Conclusion}
\label{sec:conclusion}
We show that direct access is better than prefetching for LLM inference offloading. We present \name{}, a framework that leverages the TMA to fetch weights and KV caches directly from host memory to GPU shared memory. \name{} incorporates a greedy algorithm for per-operation offloading, congestion control to prevent interconnect saturation, and TMA multicasting to eliminate read amplification. Our evaluations show that \name{} strictly outperforms prefetching, achieving near-optimal system bandwidth with performance gains of up to 3$\times$ on NVLink-C2C and 1.8$\times$ on PCIe systems.
\section{Acknowledgements}
This work is supported by gifts from Google, as well as resources provided by Lambda and the NAIRR Pilot.

\bibliographystyle{plain}
\bibliography{ref}

@inproceedings{sheng2023flexgen,
  title={Flexgen: High-throughput generative inference of large language models with a single gpu},
  author={Sheng, Ying and Zheng, Lianmin and Yuan, Binhang and Li, Zhuohan and Ryabinin, Max and Chen, Beidi and Liang, Percy and R{\'e}, Christopher and Stoica, Ion and Zhang, Ce},
  booktitle={International Conference on Machine Learning},
  pages={31094--31116},
  year={2023},
  organization={PMLR}
}

@techreport{cxlconsortium2023spec,
  title={Compute Express Link (CXL) Specification, Revision 3.1},
  author={{CXL Consortium}},
  year={2023},
  month={November},
  institution={Compute Express Link Consortium},
  url={https://computeexpresslink.org/cxl-specification/},
  note={The foundational spec enabling peer-to-peer accelerator memory pooling and fabric routing.}
}

@inproceedings{ren2021zero,
  title={$\{$Zero-offload$\}$: Democratizing $\{$billion-scale$\}$ model training},
  author={Ren, Jie and Rajbhandari, Samyam and Aminabadi, Reza Yazdani and Ruwase, Olatunji and Yang, Shuangyan and Zhang, Minjia and Li, Dong and He, Yuxiong},
  booktitle={2021 USENIX Annual Technical Conference (USENIX ATC 21)},
  pages={551--564},
  year={2021}
}

@techreport{nvidia2023gh200arch,
  title={{NVIDIA GH200 Grace Hopper Superchip Architecture}},
  author={{NVIDIA Corporation}},
  year={2023},
  institution={NVIDIA},
  type={Technical Whitepaper},
  url={https://nvdam.widen.net/s/9qz5jxz2s9/nvidia-grace-hopper-superchip-architecture-whitepaper-v1.0}
}

@article{dao2022flashattention,
  title={Flashattention: Fast and memory-efficient exact attention with io-awareness},
  author={Dao, Tri and Fu, Dan and Ermon, Stefano and Rudra, Atri and R{\'e}, Christopher},
  journal={Advances in neural information processing systems},
  volume={35},
  pages={16344--16359},
  year={2022}
}

@misc{nvidia2025rtx50series,
  title={{NVIDIA Blackwell GeForce RTX 50 Series Opens New World of AI Computer Graphics}},
  author={{NVIDIA}},
  howpublished={NVIDIA Newsroom},
  year={2025},
  month={January},
  day={6},
  url={https://nvidianews.nvidia.com/news/nvidia-blackwell-geforce-rtx-50-series-opens-new-world-of-ai-computer-graphics}
}

@techreport{nvidia2025rtxpro_blackwell_whitepaper,
  title={{NVIDIA RTX PRO Blackwell GPU Architecture}},
  author={{NVIDIA}},
  institution={NVIDIA Corporation},
  year={2025},
  url={https://www.nvidia.com/content/dam/en-zz/Solutions/design-visualization/quadro-product-literature/NVIDIA-RTX-Blackwell-PRO-GPU-Architecture-v1.0.pdf}
}

@techreport{nvidia2022hopper,
  title={{NVIDIA H100 Tensor Core GPU Architecture}},
  author={{NVIDIA Corporation}},
  year={2022},
  institution={NVIDIA},
  type={Whitepaper}
}

@article{spector2024thunderkittens,
  title={Thunderkittens: Simple, fast, and adorable ai kernels},
  author={Spector, Benjamin F and Arora, Simran and Singhal, Aaryan and Fu, Daniel Y and R{\'e}, Christopher},
  journal={arXiv preprint arXiv:2410.20399},
  year={2024}
}

@misc{cutlass3,
  title={CUTLASS 3.0: Fast Linear Algebra in CUDA C++},
  author={NVIDIA},
  howpublished={\url{https://github.com/NVIDIA/cutlass}},
  year={2023}
}

@article{shah2024flashattention,
  title={Flashattention-3: Fast and accurate attention with asynchrony and low-precision},
  author={Shah, Jay and Bikshandi, Ganesh and Zhang, Ying and Thakkar, Vijay and Ramani, Pradeep and Dao, Tri},
  journal={Advances in Neural Information Processing Systems},
  volume={37},
  pages={68658--68685},
  year={2024}
}

@article{cheng2025mirage,
  title={Mirage Persistent Kernel: A Compiler and Runtime for Mega-Kernelizing Tensor Programs},
  author={Cheng, Xinhao and Zhang, Zhihao and Zhou, Yu and Ji, Jianan and Jiang, Jinchen and Zhao, Zepeng and Xiao, Ziruo and Ye, Zihao and Huang, Yingyi and Lai, Ruihang and others},
  journal={arXiv preprint arXiv:2512.22219},
  year={2025}
}

@article{fusco2024understanding,
  title={Understanding data movement in tightly coupled heterogeneous systems: A case study with the Grace Hopper superchip},
  author={Fusco, Luigi and Khalilov, Mikhail and Chrapek, Marcin and Chukkapalli, Giridhar and Schulthess, Thomas and Hoefler, Torsten},
  journal={arXiv preprint arXiv:2408.11556},
  year={2024}
}

@inproceedings{schieffer2024harnessing,
  title={Harnessing integrated cpu-gpu system memory for hpc: a first look into grace hopper},
  author={Schieffer, Gabin and Wahlgren, Jacob and Ren, Jie and Faj, Jennifer and Peng, Ivy},
  booktitle={Proceedings of the 53rd International Conference on Parallel Processing},
  pages={199--209},
  year={2024}
}

@inproceedings{li2015evaluation,
  title={An evaluation of unified memory technology on nvidia gpus},
  author={Li, Wenqiang and Jin, Guanghao and Cui, Xuewen and See, Simon},
  booktitle={2015 15th IEEE/ACM international symposium on cluster, cloud and grid computing},
  pages={1092--1098},
  year={2015},
  organization={IEEE}
}

@inproceedings{cao2025moe,
  title={Moe-lightning: High-throughput moe inference on memory-constrained gpus},
  author={Cao, Shiyi and Liu, Shu and Griggs, Tyler and Schafhalter, Peter and Liu, Xiaoxuan and Sheng, Ying and Gonzalez, Joseph E and Zaharia, Matei and Stoica, Ion},
  booktitle={Proceedings of the 30th ACM International Conference on Architectural Support for Programming Languages and Operating Systems, Volume 1},
  pages={715--730},
  year={2025}
}

@inproceedings{chien2019performance,
  title={Performance evaluation of advanced features in CUDA unified memory},
  author={Chien, Steven and Peng, Ivy and Markidis, Stefano},
  booktitle={2019 IEEE/ACM Workshop on Memory Centric High Performance Computing (MCHPC)},
  pages={50--57},
  year={2019},
  organization={IEEE}
}

@article{gholami2024aimemorywall,
  title={{AI} and Memory Wall},
  author={Gholami, Amir and Yao, Zhewei and Kim, Sehoon and Mahoney, Michael W and Keutzer, Kurt},
  journal={IEEE Micro},
  volume={44},
  number={1},
  pages={14--24},
  year={2024},
  publisher={IEEE}
}

@inproceedings{huang2020swapadvisor,
  title={Swapadvisor: Pushing deep learning beyond the gpu memory limit via smart swapping},
  author={Huang, Chien-Chin and Jin, Gu and Li, Jinyang},
  booktitle={Proceedings of the Twenty-Fifth International Conference on Architectural Support for Programming Languages and Operating Systems},
  pages={1341--1355},
  year={2020}
}

@article{li2025fenghuang,
  title={FengHuang: Next-Generation Memory Orchestration for AI Inferencing},
  author={Li, Jiamin and Qu, Lei and Zhang, Tao and Chirkov, Grigory and Xu, Shuotao and Cheng, Peng and Zhou, Lidong},
  journal={arXiv preprint arXiv:2511.10753},
  year={2025}
}

@inproceedings{aminabadi2022deepspeed,
  title={Deepspeed-inference: enabling efficient inference of transformer models at unprecedented scale},
  author={Aminabadi, Reza Yazdani and Rajbhandari, Samyam and Awan, Ammar Ahmad and Li, Cheng and Li, Du and Zheng, Elton and Ruwase, Olatunji and Smith, Shaden and Zhang, Minjia and Rasley, Jeff and others},
  booktitle={SC22: International Conference for High Performance Computing, Networking, Storage and Analysis},
  pages={1--15},
  year={2022},
  organization={IEEE}
}

@article{kamahori2024fiddler,
  title={Fiddler: Cpu-gpu orchestration for fast inference of mixture-of-experts models},
  author={Kamahori, Keisuke and Tang, Tian and Gu, Yile and Zhu, Kan and Kasikci, Baris},
  journal={arXiv preprint arXiv:2402.07033},
  year={2024}
}

@article{xu2024pie,
  title={Pie: Pooling cpu memory for llm inference},
  author={Xu, Yi and Mao, Ziming and Mo, Xiangxi and Liu, Shu and Stoica, Ion},
  journal={arXiv preprint arXiv:2411.09317},
  year={2024}
}

@inproceedings{vijaya2025aqua,
  title={Aqua: Network-Accelerated Memory Offloading for LLMs in Scale-Up GPU Domains},
  author={Vijaya Kumar, Abhishek and Antichi, Gianni and Singh, Rachee},
  booktitle={Proceedings of the 30th ACM International Conference on Architectural Support for Programming Languages and Operating Systems, Volume 2},
  pages={48--62},
  year={2025}
}

@article{jiang2025neo,
  title={Neo: Saving gpu memory crisis with cpu offloading for online llm inference},
  author={Jiang, Xuanlin and Zhou, Yang and Cao, Shiyi and Stoica, Ion and Yu, Minlan},
  journal={Proceedings of Machine Learning and Systems},
  volume={7},
  year={2025}
}

@inproceedings{kim2025lia,
  title={Lia: A single-gpu llm inference acceleration with cooperative amx-enabled cpu-gpu computation and cxl offloading},
  author={Kim, Hyungyo and Wang, Nachuan and Xia, Qirong and Huang, Jinghan and Yazdanbakhsh, Amir and Kim, Nam Sung},
  booktitle={Proceedings of the 52nd Annual International Symposium on Computer Architecture},
  pages={544--558},
  year={2025}
}

@article{kim2024breakthrough,
  title={The breakthrough memory solutions for improved performance on llm inference},
  author={Kim, Byeongho and Cha, Sanghoon and Park, Sangsoo and Lee, Jieun and Lee, Sukhan and Kang, Shin-haeng and So, Jinin and Kim, Kyungsoo and Jung, Jin and Lee, Jong-Geon and others},
  journal={IEEE Micro},
  volume={44},
  number={3},
  pages={40--48},
  year={2024},
  publisher={IEEE}
}

@article{wolters2024memory,
  title={Memory is all you need: An overview of compute-in-memory architectures for accelerating large language model inference},
  author={Wolters, Christopher and Yang, Xiaoxuan and Schlichtmann, Ulf and Suzumura, Toyotaro},
  journal={arXiv preprint arXiv:2406.08413},
  year={2024}
}

@inproceedings{alizadeh2024llm,
  title={Llm in a flash: Efficient large language model inference with limited memory},
  author={Alizadeh, Keivan and Mirzadeh, Seyed Iman and Belenko, Dmitry and Khatamifard, S and Cho, Minsik and Del Mundo, Carlo C and Rastegari, Mohammad and Farajtabar, Mehrdad},
  booktitle={Proceedings of the 62nd Annual Meeting of the Association for Computational Linguistics (Volume 1: Long Papers)},
  pages={12562--12584},
  year={2024}
}

@inproceedings{agrawal2024taming,
  title={Taming $\{$Throughput-Latency$\}$ tradeoff in $\{$LLM$\}$ inference with $\{$Sarathi-Serve$\}$},
  author={Agrawal, Amey and Kedia, Nitin and Panwar, Ashish and Mohan, Jayashree and Kwatra, Nipun and Gulavani, Bhargav and Tumanov, Alexey and Ramjee, Ramachandran},
  booktitle={18th USENIX symposium on operating systems design and implementation (OSDI 24)},
  pages={117--134},
  year={2024}
}

@inproceedings{li2023pond,
  title={Pond: Cxl-based memory pooling systems for cloud platforms},
  author={Li, Huaicheng and Berger, Daniel S and Hsu, Lisa and Ernst, Daniel and Zardoshti, Pantea and Novakovic, Stanko and Shah, Monish and Rajadnya, Samir and Lee, Scott and Agarwal, Ishwar and others},
  booktitle={Proceedings of the 28th ACM International Conference on Architectural Support for Programming Languages and Operating Systems, Volume 2},
  pages={574--587},
  year={2023}
}

@article{gouk2023memory,
  title={Memory pooling with cxl},
  author={Gouk, Donghyun and Kwon, Miryeong and Bae, Hanyeoreum and Lee, Sangwon and Jung, Myoungsoo},
  journal={IEEE Micro},
  volume={43},
  number={2},
  pages={48--57},
  year={2023},
  publisher={IEEE}
}

@article{pearson2023interconnect,
  title={Interconnect bandwidth heterogeneity on amd mi250x and infinity fabric},
  author={Pearson, Carl},
  journal={arXiv preprint arXiv:2302.14827},
  year={2023}
}

@inproceedings{feng2023heterogeneous,
  title={Heterogeneous die-to-die interfaces: Enabling more flexible chiplet interconnection systems},
  author={Feng, Yinxiao and Xiang, Dong and Ma, Kaisheng},
  booktitle={Proceedings of the 56th Annual IEEE/ACM International Symposium on Microarchitecture},
  pages={930--943},
  year={2023}
}

@inproceedings{allen2021depth,
  title={In-depth analyses of unified virtual memory system for GPU accelerated computing},
  author={Allen, Tyler and Ge, Rong},
  booktitle={Proceedings of the International Conference for High Performance Computing, Networking, Storage and Analysis},
  pages={1--15},
  year={2021}
}

@misc{nvidia_cuda_programming_guide_tma,
  title        = {CUDA C Programming Guide: Using the Tensor Memory Accelerator (TMA)},
  author       = {{NVIDIA Corporation}},
  year         = {2024},
  url          = {https://docs.nvidia.com/cuda/cuda-programming-guide/04-special-topics/async-copies.html#using-the-tensor-memory-accelerator-tma},
  note         = {Accessed: 2026-04-15}
}

@misc{pytorch_docs_sdpa,
  title        = {PyTorch Documentation: Scaled Dot Product Attention},
  author       = {{PyTorch Contributors}},
  year         = {2024},
  howpublished = {\url{https://docs.pytorch.org/docs/stable/generated/torch.nn.functional.scaled_dot_product_attention.html}},
  note         = {Accessed: 2026-04-15}
}
% \bibliographystyle{ACM-Reference-Format}
% \bibliography{ref}
\clearpage
\appendix
\section{Formal Proof of Optimal Greedy Offload}
\label{sec:appendix_greedy_proof}
The optimization problem can be written as: 

\begin{align}
    \min_{\{x_i\}} \quad & \sum_i \frac{C_i}{\mathcal{EB}(x_i)}\label{eq:op_goal} \\
    \text{s.t.} \quad & \sum_i C_i x_i = R \sum_i C_i\label{eq:offload_constraint}, \\
    & 0 \le x_i \le 1, \quad \forall i,
\end{align}

Next, we explain why the greedy algorithm proposed in \autoref{subsec:offloading_algo} is an optimal solution to this problem.

\begin{theorem}
When $R \leq \frac{\sum_{i \in \mathcal{F}_{mem}} C_i \cdot x_i^{*}}{\sum_i C_i}$, the optimal solution allocates all offloading budget to memory-bound operations, and how the offloading budget is distributed among memory-bound operations does not affect optimality.
\end{theorem}

\begin{proof}

% \textit{Intuition:} The key intuition is that offloading cannot help compute-bound operations but helps memory-bound ones (before saturating the aggregate bandwidth), so the optimal strategy is to spend offloading budget only where it strictly reduces latency.

We prove this by contradiction. Suppose there exists an optimal solution in which some compute-bound operation $F_j \in \mathcal{F}_{comp}$ has $x_j > 0$.
Since $R \leq \frac{\sum_{i \in \mathcal{F}_{mem}} C_i x_i^*}{\sum_i C_i}$,
there must exist at least one memory-bound operation $F_k \in \mathcal{F}_{mem}$ such that $x_k < x_k^*$; otherwise, if every memory-bound operation had already reached its threshold, then the total offloaded size assigned to memory-bound operations alone would be at least $\sum_{i \in \mathcal{F}_{mem}} C_i x_i^*$, contradicting the above bound.

Now move an infinitesimal amount of offloading budget from $F_j$ to $F_k$, while keeping the total offloaded size unchanged. Let us use $\mathcal{T}(x_i) = \frac{C_i}{\mathcal{EB}(x_i)}$ to denote the latency contribution of operation $F_i$.

For the compute-bound operation $F_j$, $\mathcal{T} (x_i)$ is strictly non-decreasing in $x_j$; thus we have $T(\hat{x_{j}}) \leq T(x_j)$. For the memory-bound operation $\mathcal{T} (x_i)$ is strictly decreasing when $x_i < x_i^*$; thus we have $T(\hat{x_{k}}) < T(x_k)$. Therefore, decreasing $x_j$ and increasing $x_k$ strictly decreases the objective value, contradicting the optimality of the original solution. Hence, any optimal solution satisfies the following constraint:

\begin{align}
    \begin{cases}
        x_i \leq x_i^{*} & \forall i \in \mathcal{F}_{mem} \\
        x_i = 0 & \forall i \in \mathcal{F}_{comp}
    \end{cases}
    \label{eq:case_1_constraint}
\end{align}

% no optimal solution assigns positive offloading budget to compute-bound operations while there still exists a memory-bound operation with offloading ratio below its threshold. 

When the constraint in~\ref{eq:case_1_constraint} holds, the objective in~\ref{eq:op_goal} only depends on memory-bound operations. Rewriting the objective, we have $\min_{\{x_i\}} \sum_i \frac{C_i (1 - x_i)}{B_g}$. Since $B_g$ is a constant, this is equivalent to minimizing $\sum_i C_i (1 - x_i)$, or equivalently maximizing $\sum_i C_i x_i$. Under the global offloading constraint in~\ref{eq:offload_constraint}, $\sum_i C_i x_i$ is fixed. Therefore, the objective reduces to a constant, and any feasible solution satisfying the constraints is optimal. How the offloading budget is distributed among memory-bound operations does not affect optimality in this regime.

\end{proof}

\begin{theorem}
When $\frac{\sum_{i \in \mathcal{F}_{mem}} C_i \cdot x_i^{*}}{\sum_i C_i} < R \leq \frac{\sum_{i} C_i \cdot x_{i}^{*}} {\sum_{i} C_i}$, the optimal solution allocates offloading budget to all memory-bound operations so that they reach their peak effective bandwidth. It then allocates the remaining budget to compute-bound operations; the exact distribution among compute-bound operations does not affect optimality, as long as none of them is pushed beyond its compute-bound threshold.
\end{theorem}

\begin{proof}

Since $\frac{\mathcal{C}_i}{\mathcal{EB}(x_i)}$ is decreasing and then increasing with $x_i$ for memory-bound operations, and strictly non-decreasing with $x_i$ for compute-bound operations. The optimal solution would assign $x_i = x_i^{*}$ for all memory-bound operations, minimizing their operation latency. In addition, due to the constraint on $R \leq \frac{\sum_{i} C_i \cdot x_{i}^{*}} {\sum_{i} C_i}$, we can make sure no compute-bound operations go beyond their threshold. In this case, we can minimize $\frac{\mathcal{C}_i}{\mathcal{EB}(x_i)}$ for both memory-bound and compute-bound operations, and thus minimize the E2E latency. Thus, an optimal solution would satisfy the following constraint:

\begin{align}
    \begin{cases}
        x_i = x_i^{*} & \forall i \in \mathcal{F}_{mem} \\
        x_i \leq x_i^{*} & \forall i \in \mathcal{F}_{comp}
    \end{cases}
    \label{eq:case_2_constraint}
\end{align}

Under this constraint, we can rewrite the objective~\ref{eq:op_goal} as $\sum_{i \in \mathcal{F}_{mem}} \frac{C_i}{B_h + B_g} + \sum_{i \in \mathcal{F}_{comp}} \frac{C_i}{B_i}$, which is a constant. Hence, any feasible solution which satisfies Constraint~\ref{eq:case_2_constraint} would be optimal. How the offloading budget is distributed among compute-bound operations does not affect optimality in this regime.

\end{proof}

\begin{theorem}
When $R > \frac{\sum_{i} C_i \cdot x_{i}^{*}} {\sum_{i} C_i}$, Any feasible solution is optimal, and how the offloading budget is distributed among all operations does not affect optimality.
\end{theorem}

\begin{proof}

We first show show that for any optimal solution, it must satisfy $x_{i} \ge x_{i}^*, \forall \in \mathcal{F}$. We prove this by contradiction. Suppose there exists an optimal solution $\mathcal{S}$ in which some operation has $x_j < x_{j}^*$. Due to the constraint on offloading budget $R$, there must exist some $x_k$ with $x_k > x_{k}^*$. Then consider another solution $\hat{\mathcal{S}}$ with $\hat{x_j} = x_{j}^* > x_j$. Since we increase the offloading size of operation $j$, we can reduce the offloading size of operation $k$ while satisfying the overall offloading budget constraint for $\hat{\mathcal{S}}$. Thus, we have $\hat{x_{k}} < x_{k}$. Note that the effective bandwidth $\mathcal{EB}(x_i)$ is non-decreasing with $x_i \leq x_i^{*}$ and strictly decreasing with $x_i > x_i^{*}$. Thus, the latency of operation $j$ in solution $\hat{\mathcal{S}}$ would be no greater than its latency in solution $\mathcal{S}$, and the latency of operation $k$ in solution $\hat{\mathcal{S}}$ would be smaller than its latency in solution $\mathcal{S}$, contradicting the assumption that $\mathcal{S}$ is optimal.

Thus, we can rewrite the objective~\ref{eq:op_goal} as $\sum_{i \in \mathcal{F}_{mem}} \frac{C_i}{B_h} \cdot x_{i} + \sum_{i \in \mathcal{F}_{comp}} \frac{C_i}{B_i} \cdot x_i = \sum_{i \in \mathcal{F}_{mem}} {C_i} \cdot x_{i} + \sum_{i \in \mathcal{F}_{comp}} C_i \cdot x_i = \sum_{i \in \mathcal{F}} C_i \cdot x_i$, which is a constant, \ie the total offloading budget (due to constraint~\ref{eq:offload_constraint}).

Putting everything together, our offloading algorithm satisfies constraint~\ref{eq:case_1_constraint} in Phase~1 and constraint~\ref{eq:case_2_constraint} in Phase~2 respectively, making it optimal. In Phase~3, any feasible allocation is optimal. Therefore, the proposed offloading algorithm is optimal.

\end{proof}

\begin{figure}[t!]
    \centering
    \includegraphics[width=0.98\linewidth]{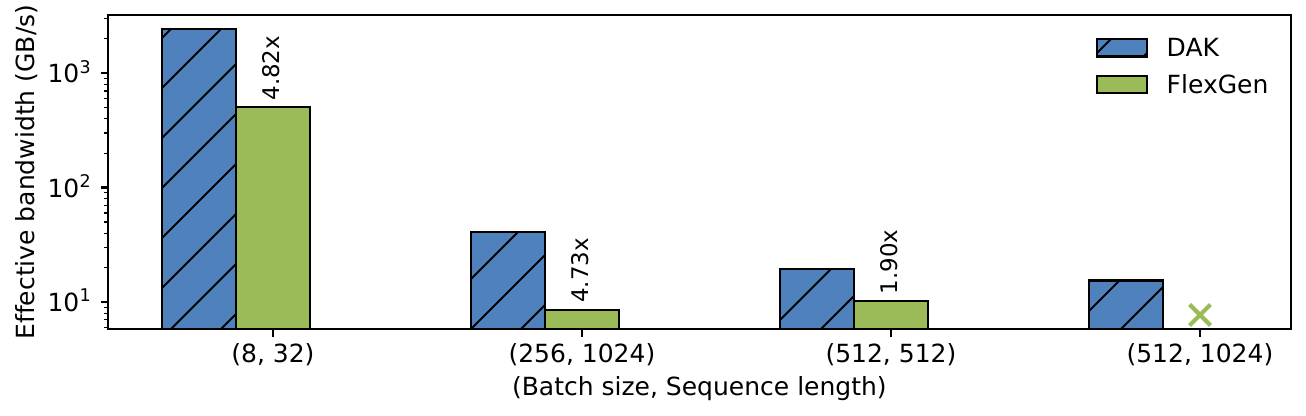}
    \footnotesize
     \begin{tabular}{ccccc}
    \hline
    bsz & prompt\_len & model  & KV cache  & global offload ratio \\ \hline
    8   & 32         & 13.3 GB      & 0.27    GB    & 0\%                \\ 
    256  & 1024       & 13.3 GB      & 141.73  GB    & 37\%                 \\ 
    512 & 512        & 13.3 GB      & 146.03  GB    & 39\%                 \\ 
    512 & 1024        & 13.3 GB      & 283.47  GB    & 67\%                 \\ \hline
    \end{tabular}
    
    \scaption{Figure shows the performance for GH200 under different configurations. The table shows the memory footprint for different configurations and the corresponding global memory offload ratio.}
     \vspace{-0.2cm}
    \label{fig:bw-config-6.7b}
\end{figure}
\section{Additional Results}
\label{appendix:addi}
\autoref{fig:bw-config-6.7b} shows the OPT 6.7B's performance for GH200 under different configurations. We vary the batch size and prompt length for OPT-6.7B. Large batch sizes and extended prompt lengths increase the total memory footprint well beyond the available GPU HBM, which
dynamically dictates the required global offload ratio. \name{} achieves 1.9x to 4x performance improvement compared to FlexGen.

\end{document}